\documentclass[12pt]{article}
\usepackage[latin1]{inputenc}
\usepackage[british]{babel}
\usepackage{cmap}
\usepackage{lmodern}

\usepackage{amssymb, amsmath, amsthm}
\usepackage[a4paper,top=25mm,bottom=25mm,left=25mm,right=25mm]{geometry}
\usepackage{ragged2e}

\usepackage{authblk} 
\usepackage{pifont}
\usepackage{graphicx}
\usepackage[usenames,dvipsnames,svgnames,table]{xcolor}
\usepackage[figuresright]{rotating}
\usepackage{xtab} 
\usepackage{longtable} 
\usepackage{multirow}
\usepackage{footnote}
\usepackage[stable]{footmisc}
\usepackage{chngpage} 
\usepackage{pdflscape} 
\usepackage{tocbibind} 

\usepackage{pgfplots}
\pgfplotsset{compat=1.14}
\pgfplotsset{every tick label/.append style={font=\footnotesize}}
\usepackage{setspace}

\makesavenoteenv{tabular}
\usepackage{tabularx}
\usepackage{booktabs}
\usepackage{threeparttable} 
\usepackage[referable]{threeparttablex} 
\newcolumntype{R}{>{\raggedleft\arraybackslash}X}
\newcolumntype{L}{>{\raggedright\arraybackslash}X}
\newcolumntype{C}{>{\centering\arraybackslash}X}
\newcolumntype{A}{>{\columncolor{gray!25}}C}
\newcolumntype{a}{>{\columncolor{gray!25}}c}

\newlength{\tablen}

\usepackage{dcolumn} 
\newcolumntype{.}{D{.}{.}{-1}}

\usepackage{tikz}
\usetikzlibrary{arrows, calc, matrix, patterns, positioning, trees}
\usepackage[semicolon]{natbib}
\usepackage[hyphens]{url}
\usepackage{hyperref} 
\hypersetup{
  colorlinks   = true,    		
  urlcolor     = blue,    		
  linkcolor    = blue,			
  citecolor    = ForestGreen,		  	
}
\usepackage{microtype}
\usepackage[justification=centerfirst]{caption} 

\usepackage[labelformat=simple]{subcaption}

\DeclareCaptionLabelFormat{parenthesis}{(#2)}
\makeatletter
\renewcommand\p@subfigure{\arabic{figure}.}
\makeatother

\DeclareCaptionLabelFormat{parenthesis}{(#2)}
\makeatletter
\renewcommand\p@subtable{\arabic{table}.}
\makeatother

\usepackage{enumitem}

\setlist[itemize]{leftmargin=2.5\parindent}
\setlist[enumerate]{leftmargin=2.5\parindent}

\newcommand\Pair[3]{%
  \begin{tabular}{|>{\centering\arraybackslash}m{0.75cm}|>{\centering\arraybackslash}m{1.5cm}|}
  \hline
  \multirow{2}{*}{#1} & #2 \\ \cline{2-2}
   & #3 \\
  \hline
  \end{tabular}%
}

\theoremstyle{plain}

\newtheorem{proposition}{Proposition}[section]

\theoremstyle{definition}

\newtheorem{definition}{Definition}[section]
\newtheorem{example}{Example}[section]

\theoremstyle{remark}
\newtheorem{remark}{Remark}[section]

\def\keywords{\vspace{.5em} 
{\noindent \textit{Keywords}: }}

\def\AMS{\vspace{.5em} 
{\noindent \textbf{\emph{MSC} class}: }}

\def\JEL{\vspace{.5em} 
{\noindent \textbf{\emph{JEL} classification number}: }}

\author{\href{https://sites.google.com/view/laszlocsato}{L\'aszl\'o Csat\'o}\thanks{~E-mail: \emph{laszlo.csato@sztaki.hu}} }
\affil{Institute for Computer Science and Control (SZTAKI) \\
E\"otv\"os Lor\'and Research Network (ELKH) \\
Laboratory on Engineering and Management Intelligence \\
Research Group of Operations Research and Decision Systems}
\affil{Corvinus University of Budapest (BCE) \\
Department of Operations Research and Actuarial Sciences}
\affil{Budapest, Hungary}
\title{How to design a multi-stage tournament \\ when some results are carried over?}

\date{\today}

\def\Dedication{ 
{\noindent $\mathfrak{Es}$ $\mathfrak{ist}$ $\mathfrak{dabei}$ $\mathfrak{selbst}$ $\mathfrak{die}$ $\mathfrak{historische}$ $\mathfrak{Wahrheit}$ $\mathfrak{eine}$ $\mathfrak{Nebensache}$, $\mathfrak{ein}$ $\mathfrak{erfundenes}$ $\mathfrak{Beispiel}$ $\mathfrak{k\ddot{o}nnte}$ $\mathfrak{auch}$ $\mathfrak{dienen}$; $\mathfrak{nur}$ $\mathfrak{haben}$ $\mathfrak{historische}$ $\mathfrak{immer}$ $\mathfrak{den}$ $\mathfrak{Vorzug}$, $\mathfrak{praktischer}$ $\mathfrak{zu}$ $\mathfrak{sein}$ $\mathfrak{und}$ $\mathfrak{den}$ $\mathfrak{Gedanken}$, $\mathfrak{welchen}$ $\mathfrak{sie}$ $\mathfrak{erl\ddot{a}utern}$, $\mathfrak{dem}$ $\mathfrak{praktischen}$ $\mathfrak{Leben}$ $\mathfrak{selbst}$ $\mathfrak{n\ddot{a}her}$ $\mathfrak{zu}$ $\mathfrak{f\ddot{u}hren.}$}\footnote{~
``\emph{Historical correctness is a secondary consideration; a case invented might also serve the purpose as well, only historical ones are always to be preferred, because they bring the idea which they illustrate nearer to practical life.}'' (Source: Carl von Clausewitz: \emph{On War}, Book 2, Chapter 6 [On Examples], translated by Colonel James John Graham, London, N. Tr\"ubner, 1873. \url{http://clausewitz.com/readings/OnWar1873/TOC.htm})}
\vspace{0.25cm}

\flushright
\noindent (Carl von Clausewitz: \emph{Vom Kriege})

\vspace{1cm} 
\justify }

\begin{document}

\newgeometry{top=15mm,bottom=25mm,left=25mm,right=25mm}
\maketitle
\thispagestyle{empty}
\Dedication

\begin{abstract}
\noindent
The paper discusses the strategy-proofness of sports tournaments with multiple group stages, where the results of matches already played in the previous round against teams in the same group are carried over. These tournaments, widely used in handball and other sports, are shown to be incentive incompatible in the sense that a team can be strictly better off by not exerting full effort in a game. Historical examples are presented when a team was ex ante disinterested in winning by a high margin. We propose a family of incentive compatible designs. Their main characteristics are compared to the original format via simulations. Carrying over half of the points scored in the previous round turns out to be a promising policy.


\AMS{62F07, 68U20, 91A80, 91B14}

\JEL{C44, C63, D71, Z20}
\end{abstract}

\clearpage
\restoregeometry

\section{Introduction} \label{Sec1}

Strategy-proofness is a central issue in sports where all contestants are familiar with the high-stake decisions involved and behave as strategic actors. Consequently, a tournament design should provide the players with the appropriate incentives to perform \citep{Szymanski2003}.

Although sporting applications of operations research proliferate in the academic world \citep{Wright2014}, the scientific analysis of sports ranking rules from the perspective of incentive compatibility has started only recently. \citet{KendallLenten2017} give probably the first comprehensive review of sports regulations with unexpected consequences. Their examples uncover three possible situations in which a team might prefer losing a game to winning it:
(1) when it might gain advantages in the next season;
(2) when a lower-ranked team can still qualify and might face a favoured competitor in a later stage of the tournament;
(3) when a team is strictly better off by losing due to an ill-constructed tournament design.

The classical example for the first situation arises from the reverse order applied in the traditional set-up of player drafts that aims to increase competitive balance over time. Hence, if a team is still certainly eliminated from the play-off, a perverse incentive is created to tank in the later games \citep{TaylorTrogdon2002, PriceSoebbingBerriHumphreys2010, Fornwagner2019}.

The second situation occurred, for instance, in \href{https://en.wikipedia.org/wiki/Badminton_at_the_2012_Summer_Olympics_\%E2\%80\%93_Women\%27s_doubles}{Badminton at the 2012 Summer Olympics---Women's doubles} \citep[Section~3.3.1]{KendallLenten2017}, and has inspired some works addressing the strategic manipulation problem with game-theoretical tools \citep{Pauly2014, Vong2017}.

However, in the first case, the rules are deliberately devised to support underdogs, and in the second case, the team gains only in expected terms.
The current paper discusses the most serious third situation when the tournament rules allow a team to certainly benefit from a weaker performance. Therefore, a design is called \emph{strategy-proof} in the following if this possibility is excluded. We do not deal with other forms of strategic manipulation like collusion and shirking.

Probably the first academic paper studying the problem of such misaligned incentives is \citet{DagaevSonin2018}. The authors prove that tournament systems, consisting of multiple round-robin and knockout tournaments with noncumulative prizes, are generically incentive incompatible. Recent qualifications for the UEFA (Union of European Football Associations) European Championships have also been shown to be vulnerable to manipulation \citep{Csato2018b, Csato2020c}, including the case when both teams should avoid winning to advance to the next stage \citep{Csato2020d}.

Here it will be revealed that sports tournaments with multiple group stages, in which some (but not all) match results from the previous round are carried over to the next round, suffer from incentive incompatibility. In particular, teams initially play a preliminary round-robin stage and the top teams qualify for a second (main) round-robin stage. In this second stage, some groups are ``merged'', that is, two teams qualifying from the same group in the preliminary stage will be in the same main round-robin group. In order to reduce the number of matches in the tournament, these teams do not play another game in the main stage against each other but its result is carried over from the preliminary stage.

\begin{table}[t]
\centerline{
\begin{threeparttable}
\caption{Recent handball tournaments with multiple group stages and results that are carried over}
\label{Table1}
\rowcolors{1}{gray!20}{}
\begin{tabularx}{1.1\linewidth}{Llc ccc ccc} \toprule \hiderowcolors
     \multirow{2}{*}{Tournament} & \multirow{2}{*}{Year(s)} & \multirow{2}{*}{Type} & \multicolumn{3}{c}{Preliminary round}       & \multicolumn{3}{c}{Main round} \\ 
     & & & Gr.($k$)   & Teams & Q     & Gr.($\ell$)   & Teams & Q \\ \bottomrule \showrowcolors
    \href{https://en.wikipedia.org/wiki/European_Men\%27s_Handball_Championship}{EHF Euro Men} & \href{https://en.wikipedia.org/wiki/2002_European_Men\%27s_Handball_Championship}{2002}--\href{https://en.wikipedia.org/wiki/2018_European_Men's_Handball_Championship}{2018} & S     & 4     & 4     & 3     & 2     & 6     & 2 \\
    \href{https://en.wikipedia.org/wiki/European_Men\%27s_Handball_Championship}{EHF Euro Men} & \href{https://en.wikipedia.org/wiki/2020_European_Men's_Handball_Championship}{2020}-- & S     & 6     & 4     & 2     & 2     & 6     & 2 \\
    \href{https://en.wikipedia.org/wiki/European_Women\%27s_Handball_Championship}{EHF Euro Women} & \href{https://en.wikipedia.org/wiki/2002_European_Women's_Handball_Championship}{2002}-- & S     & 4     & 4     & 3     & 2     & 6     & 2 \\
    \href{https://en.wikipedia.org/wiki/Women\%27s_EHF_Champions_League}{EHF Women's CL} & \href{https://en.wikipedia.org/wiki/2013\%E2\%80\%9314_UEFA_Women\%27s_Champions_League}{2013/14}--\href{https://en.wikipedia.org/wiki/2019\%E2\%80\%9320_UEFA_Women\%27s_Champions_League}{2019/20} & D     & 4     & 4     & 3     & 2     & 6     & 4 \\
    \href{https://en.wikipedia.org/wiki/IHF_World_Men\%27s_Handball_Championship}{IHF World Men} & \href{https://en.wikipedia.org/wiki/2003_World_Men\%27s_Handball_Championship}{2003}  & S     & 4     & 6     & 4     & 4     & 4     & 1 \\
    \href{https://en.wikipedia.org/wiki/IHF_World_Men\%27s_Handball_Championship}{IHF World Men} & \href{https://en.wikipedia.org/wiki/2005_World_Men\%27s_Handball_Championship}{2005}, \href{https://en.wikipedia.org/wiki/2009_World_Men\%27s_Handball_Championship}{2009}-\href{https://en.wikipedia.org/wiki/2011_World_Men\%27s_Handball_Championship}{2011}, \href{https://en.wikipedia.org/wiki/2019_World_Men\%27s_Handball_Championship}{2019}  & S     & 4     & 6     & 3     & 2     & 6     & 2 \\
    \href{https://en.wikipedia.org/wiki/IHF_World_Men\%27s_Handball_Championship}{IHF World Men} & \href{https://en.wikipedia.org/wiki/2007_World_Men\%27s_Handball_Championship}{2007}  & S     & 6     & 4     & 2     & 2     & 6     & 4 \\
    \href{https://hu.wikipedia.org/wiki/F\%C3\%A9rfi_k\%C3\%A9zilabda-vil\%C3\%A1gbajnoks\%C3\%A1g}{IHF World Men} & \href{https://hu.wikipedia.org/wiki/2021-es_f\%C3\%A9rfi_k\%C3\%A9zilabda-vil\%C3\%A1gbajnoks\%C3\%A1g}{2021}  & S     & 8     & 4     & 3     & 4     & 6     & 2 \\
    \href{https://en.wikipedia.org/wiki/IHF_World_Women\%27s_Handball_Championship}{IHF World Women} & \href{https://en.wikipedia.org/wiki/2003_World_Women\%27s_Handball_Championship}{2003}-\href{https://en.wikipedia.org/wiki/2005_World_Women\%27s_Handball_Championship}{2005}, \href{https://en.wikipedia.org/wiki/2009_World_Women\%27s_Handball_Championship}{2009}, \href{https://en.wikipedia.org/wiki/2019_World_Women\%27s_Handball_Championship}{2019} & S     & 4     & 6     & 3     & 2     & 6     & 2 \\
    \href{https://en.wikipedia.org/wiki/IHF_World_Women\%27s_Handball_Championship}{IHF World Women} & \href{https://en.wikipedia.org/wiki/2007_World_Women\%27s_Handball_Championship}{2007}  & S     & 6     & 4     & 2     & 2     & 6     & 4 \\ \bottomrule
\end{tabularx}
\begin{tablenotes} \footnotesize
\item
Notes: S = single round-robin (in groups); D = double round-robin (in groups); Gr. = Number of groups in the preliminary and main round, respectively, which will be denoted by $k$ and $\ell$ in the theoretical model of Section~\ref{Sec3}; Teams = Number of teams in each group of the preliminary and main round, respectively; Q = Number of teams qualified from each group of the preliminary and main round, respectively.
\item
Abbreviations: EHF Euro Men (Women) = European Men's (Women's) Handball Championship; EHF Women's CL = Women's EHF Champions League; IHF World Men (Women) = IHF World Men's (Women's) Handball Championship.
\end{tablenotes}
\end{threeparttable}
}
\end{table}

This format is widely used in handball as Table~\ref{Table1} demonstrates. All these tournaments contain two group stages, and the number of qualified teams in the main round shows the number of teams that have a chance to win the tournament at the end of this phase (see the last column).
Tournaments with multiple group stages and carried over results are also used in other sports, for instance, in basketball (\href{https://en.wikipedia.org/wiki/EuroBasket_2013}{EuroBasket 2013}), cricket (\href{https://en.wikipedia.org/wiki/2007_Cricket_World_Cup}{2007 Cricket World Cup}) \citep{ScarfYusofBilbao2009}, and volleyball (\href{https://en.wikipedia.org/wiki/2014_FIVB_Volleyball_Men\%27s_World_Championship}{2014 FIVB Volleyball Men's World Championship}).
Although \href{https://en.wikipedia.org/wiki/1999\%E2\%80\%932000_UEFA_Champions_League}{1999/2000 UEFA Champions League}, as well as the following three seasons of this tournament, included two subsequent group stages, no results were carried over to the second group stage.

First, we present a real-world handball match where a team had an incentive not to win by a high margin.
Second, this particular tournament design is verified to violate strategy-proofness in general.
Finally, an incentive compatible mechanism is provided, namely, to carry over a monotonic transformation of all results from the previous round, regardless that some matches were played against teams already eliminated from the tournament. According to computer simulations, carrying over half of all points scored in the previous round essentially does not affect the selective ability and the competitive balance of the tournament, while it guarantees strategy-proofness and even reduces the influence of seeding the teams into pots before the draw of the groups.
Our suggestion has been discussed in a recent collection of academic work proposing rule change ideas \citep{LentenKendall2021}.

The main contributions can be summarised as follows:
(1) the incentive incompatibility of sports tournaments with multiple group stages is proved by a mathematical model;
(2) a real-world example is presented to show that this is not only an irrelevant issue in practice, and a third, innocent team might suffer from the unsportsmanlike act of a team in a match;
(3) inspired by a policy applied in some European association football leagues, a viable strategy-proof alternative is suggested.
These results can be especially useful for sports administrators.

The mathematical framework for multi-stage tournaments somewhat overlaps with the models of \citet{Csato2020c} and \citet{Csato2020d}. But both works investigate the problem of qualification systems where some teams playing in different groups should be compared, and they do not deal with the next stage of the tournament and do not consider the possibility of carrying over some points at all.
The type of sports tournaments analysed here is completely new in the literature discussing incentive (in)compatibility. Unfortunately, contrary to \citet{DagaevSonin2018} and \citet{Csato2020c}, no straightforward strategy-proof mechanism can be found for the original format of tournaments with multiple group stages and carried over results.
Thus alternative incentive compatible solutions should be proposed and investigated by Monte-Carlo simulations. The simulation methodology is imported from \citet{Csato2021b}, a paper that aims only to compare certain tournament designs without addressing incentive compatibility.

The rest of the paper proceeds as follows.
Section~\ref{Sec2} brings an example from handball, where the unfair behaviour of a team could lead to the elimination of a third team.
Section~\ref{Sec3} builds a theoretical model to prove that a standard tournament with multiple group stages violates strategy-proofness.
Section~\ref{Sec4} proposes a family of incentive compatible mechanisms for organising these tournaments and explores their characteristics with respect to selective ability and competitive balance via Monte-Carlo simulations.
Finally, Section~\ref{Sec5} concludes.

\section{A real-world example of misaligned incentives} \label{Sec2}

The \href{https://en.wikipedia.org/wiki/European_Men\%27s_Handball_Championship}{European Men's Handball Championship} is a biennial competition for the senior men's national handball teams of Europe since 1994, organised by the EHF (European Handball Federation), the umbrella organization for European handball. The \href{https://en.wikipedia.org/wiki/2014_European_Men's_Handball_Championship}{11th European Men's Handball Championship (EHF Euro 2014)} was held in Denmark between 12 and 26 January 2014.
In its preliminary round, the sixteen national teams were divided into four groups (A--D) to play in a round-robin format. The top three teams in each group qualified to the main round: teams from Groups A and B composed the first main round group X, while teams from Groups C and D composed the second main round group Y. The main round groups were also organised in a round-robin format, but all matches (consequently, results and points) played in the preliminary round between teams that were in the same main round group, were kept and remained valid for the ranking of the main round.
Figure~\ref{Fig_A1} in the Appendix gives an overview of this tournament design.

In the groups of the preliminary and main rounds, two points were awarded for a win, one point for a draw, and zero points for a defeat. Teams were ranked by adding up their number of points.
If two or more teams had an equal number of points, the following tie-breaking criteria were used after the completion of all group matches \citep[Articles~9.11--9.12 and 9.23--9.24]{EHF2021}:
\noindent \emph{
\begin{enumerate}[label=\alph*)]
\item
Higher number of points obtained in the group matches played amongst the teams in question;
\item
Superior goal difference from the group matches played amongst the teams in question;
\item
Higher number of goals scored in the group matches played amongst the teams in question;
\item
Superior goal difference from all group matches (achieved by subtraction);
\item
Higher number of goals scored in all group matches.
\end{enumerate}
}

\begin{table}[t]
\centering
\caption{11th European Men's Handball Championship (EHF Euro 2014), Group C}
\label{Table2}
\begin{subtable}{\linewidth}
\centering
\caption{Match results}
\label{Table2a}
\rowcolors{1}{}{gray!20}
    \begin{tabularx}{0.9\linewidth}{lLLc} \toprule
    Date  & First team & Second team & Result \\ \bottomrule \showrowcolors
    13 January 2014, 18:00 & Serbia & Poland & 20-19 \\
    13 January 2014, 20:15 & France & Russia & 35-28 \\ \hline
    15 January 2014, 18:00 & Russia & Serbia & 27-25 \\
    15 January 2014, 20:15 & Poland & France & 27-28 \\ \hline
    17 January 2014, 18:00 & Poland & Russia & to be played \\
    17 January 2014, 20:15 & Serbia & France & to be played \\ \bottomrule
    \end{tabularx}
\end{subtable}

\vspace{0.5cm}
\begin{threeparttable}
\begin{subtable}{\linewidth}

\centering
\caption{Standing after two matchdays}
\label{Table2b}
\rowcolors{1}{}{gray!20}
    \begin{tabularx}{\linewidth}{Cl CCCC CC >{\bfseries}C} \toprule \showrowcolors
    Pos   & Team   & W     & D     & L     & GF    & GA    & GD    & Pts \\ \bottomrule
    1     & France & 2     & 0     & 0     & 63    & 55    & 8     & 4 \\
    2     & Serbia & 1     & 0     & 1     & 45    & 46    & -1    & 2 \\
    3     & Russia & 1     & 0     & 1     & 55    & 60    & -5    & 2 \\ \toprule
    4     & Poland & 0     & 0     & 2     & 46    & 48    & -2    & 0 \\ \bottomrule    
    \end{tabularx}
\end{subtable}
\begin{tablenotes} \footnotesize
\item
Notes: Pos = Position; W = Won; D = Drawn; L = Lost; GF = Goals for; GA = Goals against; GD = Goal difference; Pts = Points. \\
All teams have played 2 matches. \\
The top three teams qualify for the main round.
\end{tablenotes}
\end{threeparttable}
\end{table}

A strange situation emerged in Group C of the preliminary round, which requires further investigation. On 16 January 2014, each team in the group had one more game to play.
Table~\ref{Table2} shows the known results and the preliminary standing.

Consider the possible scenarios from the perspective of Poland.
It is certainly eliminated if it does not win against Russia.
Poland carries over $0$ points, $46$ goals for and $48$ goals against to the main round if it wins against Russia and Serbia plays at least a draw against France because then Russia will be eliminated as the fourth team of the group.
On the other hand, if Poland wins by $x$ goals against Russia and Serbia loses, there will be three teams with $2$ points, which obtained $2$ points in the group matches played among them. Consequently, further tie-breaking criteria should be applied: Poland, Russia, and Serbia will have head-to-head goal differences of $x-1$, $2-x$, and $-1$, respectively.

$x-1 > -1$ implies that Poland will qualify.
Serbia is eliminated as being the fourth team if $1 \leq x \leq 2$.
Russia and Serbia have the same head-to-head goal difference if $x=3$, hence the number of goals scored against the three teams with $2$ points breaks the tie. It is $45$ for Serbia and at least $27$ for Russia, thus Russia qualifies if it scores at least $19$ goals against Poland (if Poland vs.\ Russia is 21-18, then the ranking will depend on the result of Serbia vs.\ France).
If $x \geq 4$, then Serbia has a better head-to-head goal difference than Russia, thus Serbia qualifies, and Russia is eliminated.

\begin{figure}[t]
\centering

\begin{tikzpicture}[scale=1,auto=center, transform shape, >=triangle 45]
\tikzstyle{every node}=[draw,align=center];
  \node (C1) at (0,6) {Does Poland win by $x$ goals against Russia?};
  \node (C2) at (-4,3) {Does Serbia lose against France?};
  \node (C3) at (-4,0) {Is $x$ less than 4?};
  \node (C4) at (-4,-3) {Is $x$ equal to 3?};

\tikzstyle{every node}=[align=center];
  \node (O1) at (4,3) {\textcolor{red}{Poland is eliminated}};
  \node (O2) at (4,0) {\textcolor{blue}{Poland qualifies with 0 points}};
  \node (O3) at (4,-3) {\textcolor{blue}{Poland qualifies with 0 points}};
  \node (O4) at (-4,-6) {\textcolor{blue}{Poland qualifies with 0 or 2 points}};
  \node (O5) at (4,-6) {\textcolor{ForestGreen}{\textbf{Poland qualifies with 2 points}}};
  
\tikzstyle{every node}=[align=center];  
  \draw [->,line width=1pt] (C1) -- (C2)  node [midway,above left] {Yes};
  \draw [->,line width=1pt] (C1) -- (O1)  node [midway,above right] {No};  
  \draw [->,line width=1pt] (C2) -- (C3)  node [midway,left] {Yes};
  \draw [->,line width=1pt] (C2) -- (O2)  node [midway,above right] {No};
  \draw [->,line width=1pt] (C3) -- (C4)  node [midway,left] {Yes};
  \draw [->,line width=1pt] (C3) -- (O3)  node [midway,above right] {No};
  \draw [->,line width=1pt] (C4) -- (O4)  node [midway,left] {Yes};
  \draw [->,line width=1pt] (C4) -- (O5)  node [midway,above right] {No};
\end{tikzpicture}

\captionsetup{justification=centering}
\caption{Possible scenarios before the last matchday of Group C in the \\ 2014 European Men's Handball Championship from the perspective of Poland}
\label{Fig1}
\end{figure}


Figure~\ref{Fig1} overviews all scenarios from the perspective of Poland.
To summarise, if Poland wins, it carries over its result against Russia ($2$ points) or Serbia ($0$ points) to the main round, therefore Poland has every incentive to qualify together with Russia.
Consequently, it is \emph{ex ante} unfavourable for Poland to win by more than three goals against Russia because this scenario yields no gain in the main round but may lead to a loss of $2$ points if Serbia is defeated by France. However, Russia does not have similar problems with its incentives, for example, it is clearly better off by a smaller defeat compared to a greater one.

In fact, Poland vs.\ Russia was 24-22 and Serbia vs.\ France was 28-31, hence France, Poland, and Russia qualified for the main round with $4$, $2$, and $0$ points, respectively.
The result of Poland vs.\ Russia was 10-14 after $30$ minutes (half-time), while the match stood at 21-16 in the $48$th, 22-17 in the $50$th, and 23-18 in the $52$nd minute \citep{EHF2014b}.\footnote{~A video of the match Poland vs.\ Russia is available at \url{https://www.youtube.com/watch?v=dQvEAzyBgGo}.}

These events, perhaps influenced by the misaligned incentives of Poland, led to the elimination of a third, innocent team, Serbia, which makes the example especially worrying. The situation could not have been improved by playing the last group matches simultaneously because Poland's (weakly) dominant strategy was independent of the result of the game played later.
This seems to be a persuading argument against the rules of the 11th European Men's Handball Championship (EHF Euro 2014).

\section{The model} \label{Sec3}

Now we build a model of a tournament consisting of round-robin preliminary and main rounds, where the matches played in the preliminary round against teams that qualified to the same main round group are carried over. This design will be proved to violate strategy-proofness, that is, it allows for misaligned incentives.
Our notations follow \citet{Csato2020c} in certain details since the qualification system discussed there is also based on round-robin groups.

\begin{definition} \label{Def31}
\emph{Round-robin group}:
The pair $(X,R)$ is a \emph{round-robin group} where
\begin{itemize}
\item
$X$ is a finite set of at least two teams;
\item
the \emph{ranking method} $R$ associates a strict order $R(v)$ on the set $X$ for any function $v: X \times X \to \left\{ \left( v_1; v_2 \right): v_1,v_2 \in \mathbb{N} \right\} \cup \{ \text{---} \} \cup \{ \otimes \}$ such that $v(x,y) = \text{---}$ if $x=y$ and $v(x,y) = \otimes$ implies $v(y,x) = \otimes$.
\end{itemize}
\end{definition}

Function $v$ describes game results with the number of goals scored by the first and second team, respectively. It contains the possibility that some matches between the teams remain to be played, denoted by the symbol $\otimes$.

Definition~\ref{Def31} can describe a home-and-away round-robin tournament where any two teams play each other once at home and once at away.
Consider the notation $v(u,w) = \left( v_1(u,w); v_2(u,w) \right)$ where the first team $u$ plays at home.
It is said that team $x$ wins over team $y$ if $v_1(x,y) > v_2(x,y)$ (home) or $v_1(y,x) < v_2(y,x)$ (away), team $x$ loses to team $y$ if $v_1(x,y) < v_2(x,y)$ (home) or $v_1(y,x) > v_2(y,x)$ (away) and team $x$ draws against team $y$ if $v_1(x,y) = v_2(x,y)$ or $v_1(y,x) = v_2(y,x)$.

Let $(X,R)$ be a round-robin group, $x,y \in X$, $x \neq y$ be two teams, and $v$ be a set of results.
$x$ is ranked higher (lower) than $y$ if and only if $x$ is preferred to $y$ by $R(v)$, that is, $x \succ_{R(v)} y$ ($x \prec_{R(v)} y$).

The ranking is usually based on the number of points scored.

\begin{definition} \label{Def32}
\emph{Number of points}:
Let $(X,R)$ be a round-robin group, $v$ be a set of results, $x \in X$ be a team, and $\alpha > \beta > \gamma$ be three parameters.
Denote by $N_v^w(x)$ the number of wins, by $N_v^d(x)$ the number of draws, and by $N_v^\ell(x)$ the number of losses of team $x$, respectively.
The \emph{number of points} of team $x$ is $s_v(x) = \alpha N_v^w(x) + \beta N_v^d(x) + \gamma N_v^\ell(x)$.
\end{definition}

In other words, a win gives $\alpha$, a draw gives $\beta$, and a loss gives $\gamma$ points.

\begin{remark} \label{Rem31}
With a slight abuse of notation, it is assumed in the following that the ranking method $R$ determines the values $\alpha > \beta > \gamma$ for any round-robin group $(X,R)$.
\end{remark}

The number of points does not necessarily induce a strict order on the set of teams, hence some tie-breaking rules should be introduced.

\begin{definition} \label{Def33}
\emph{Goal difference}:
Let $(X,R)$ be a round-robin group, $v$ be a set of results, and $x \in X$ be a team.
The \emph{goal difference} of team $x$ is
\[
gd_v(x) = \sum_{y \in X \setminus \{ x \}, \, v(x,y) \neq \otimes} \left[ v_1(x,y) - v_2(x,y) \right] + \sum_{y \in X \setminus \{ x \}, \, v(x,y) \neq \otimes} \left[ v_2(y,x) - v_1(y,x) \right].
\]
\end{definition}

Goal difference is the number of goals scored by team $x$ minus the number of goals conceded by team $x$.

\begin{definition} \label{Def34}
\emph{Head-to-head results}:
Let $(X,R)$ be a round-robin group, $v$ be a set of results, and $x \in X$ be a team.
Denote by $L \subseteq X \setminus \{ x \}$ a set of teams. \\
The \emph{head-to-head number of points} of team $x$ with respect to $L$ is
\begin{eqnarray*}
s_v^L(x) & = & \alpha \left( | \left\{ y \in L: v_1(x,y) > v_2(x,y) \right\} | + | \left\{ y \in L: v_1(y,x) < v_2(y,x) \right\} | \right) + \\
& & + \beta \left( | \left\{ y \in L: v_1(x,y) = v_2(x,y) \right\} | + | \left\{ y \in L: v_1(y,x) = v_2(y,x) \right\} | \right) + \\
& & + \gamma \left( | \left\{ y \in L: v_1(x,y) < v_2(x,y) \right\} | + | \left\{ y \in L: v_1(y,x) > v_2(y,x) \right\} | \right).
\end{eqnarray*}
The \emph{head-to-head goal difference} of team $x$ with respect to $L$ in $(X,v)$ is
\[
gd_v^L(x) = \sum_{y \in L} \left[ v_1(x,y) - v_2(x,y) \right] + \sum_{y \in L} \left[ v_2(y,x) - v_1(y,x) \right].
\]
\end{definition}

In accordance with \citet[Articles~9.11--9.12 and 9.23--9.24]{EHF2021}, head-to-head results are calculated only for complete round-robin groups where all matches have already been played.

\begin{definition} \label{Def35}
\emph{Head-to-head domination}:
Let $(X,R)$ be a round-robin group, $v$ be a set of results, and $x, y \in X$ be two teams such that $s_v(x) = s_v(y)$. Denote by $L$ the set of teams that have scored the same number of points as teams $x$ and $y$.
Team $x$ \emph{head-to-head dominates} team $y$ if one of the following holds:
\begin{itemize}
\item
$s_v^L(x) > s_v^L(y)$;
\item
$s_v^L(x) = s_v^L(y)$ and $gd_v^L(x) > gd_v^L(y)$.
\end{itemize}
\end{definition}

Therefore, if two teams have the same number of points, then one head-to-head dominates the other if:
(a) it has scored more points against all teams with the same number of points (the first condition); or
(b) it has scored the same number of points against all teams with the same number of points but has a superior goal difference against them (the second condition). See the analogy to \citet[Articles~9.11--9.12 and 9.23--9.24]{EHF2021}.

\begin{definition} \label{Def36}
\emph{Monotonicity of the ranking in a round-robin group}:
Let $(X,R)$ be a round-robin group.
Its ranking method is called \emph{monotonic} if for any set of results $v$ and for any teams $x,y \in X$:
\begin{enumerate}
\item \label{eq1}
$s_v(x) > s_v(y)$ implies $x \succ_{R(v)} y$;
\item \label{eq2}
$s_v(x) = s_v(y)$, $gd_v(x) > gd_v(y)$, and $x$ head-to-head dominates $y$ imply $x \succ_{R(v)} y$.
\end{enumerate}
\end{definition}

Monotonicity requires that
(a) a team should be ranked higher if it has a greater number of points (criterion~\ref{eq1}); and
(b) a team should be ranked higher compared to any other with the same number of points, an inferior goal difference, and worse head-to-head results against all teams with the same number of points (criterion~\ref{eq2}).

Monotonicity does not necessarily result in a strict ranking. The complexity of Definition~\ref{Def36} is due to cover the two different tie-breaking concepts, goal difference, and head-to-head results. For example, in association football, FIFA currently uses the former, while UEFA applies the latter rule.

\begin{definition} \label{Def37}
\emph{Preliminary round}:
The \emph{preliminary round} $\mathcal{P}$ consists of $k$ groups of round-robin tournaments $(X^1,R^1)$, $(X^2,R^2)$, \dots , $(X^k,R^k)$ such that $X^i \cap X^h = \emptyset$ for any $h \neq i$, $1 \leq h,i \leq k$.
\end{definition}

\begin{definition} \label{Def38}
\emph{Main round}:
The \emph{main round} $\mathcal{M}$ consists of $\ell$ groups of round-robin tournaments $(Y^1,S^1)$, $(Y^2,S^2)$, \dots , $(Y^\ell,S^\ell)$ such that $Y^j \cap Y^h = \emptyset$ for any $j \neq h$, $1 \leq h,j \leq \ell$.
\end{definition}

\begin{definition} \label{Def39}
\emph{Qualification rule}:
Let $\mathcal{P}$ be a preliminary round and $\mathcal{M}$ be a main round.
For any set of results $V = \left\{ v^1, v^2, \dots , v^k \right\}$ in the preliminary round such that $v^i(x,y) \neq \otimes$ for all $x,y \in X^i$ and $1 \leq i \leq k$, a \emph{qualification rule} $\mathcal{Q}$ associates the sets $Y^1, Y^2, \dots Y^\ell$ and the set of results $W = \left\{ w^1, w^2, \dots , w^\ell \right\}$ in the main round groups.
\end{definition}

Thus the qualification rule determines the composition of the groups in the main round and the set of results carried over from the preliminary round on the basis of all results in the preliminary round, that is, after all matches have been played there.

Team $x \in X^i$ is said to be \emph{qualified} to the main round if $x \in \cup_{j=1}^\ell Y^j$.

\begin{definition} \label{Def310}
\emph{Tournament with multiple group stages}:
A \emph{tournament with multiple group stages} is a triple $(\mathcal{P}, \mathcal{M}, \mathcal{Q})$ consisting of the preliminary round $\mathcal{P}$, the main round $\mathcal{M}$, and the qualification rule $\mathcal{Q}$.\footnote{~The definition contains only two group stages in order to simplify the notations. All definitions and statements can be easily generalised to more group stages, which is left to the reader.}
\end{definition}

It is natural to restrict our attention to a reasonable subset of tournaments.

\begin{definition} \label{Def311}
\emph{Regularity of a tournament with multiple group stages}:
Let $(\mathcal{P}, \mathcal{M}, \mathcal{Q})$ be a tournament with multiple group stages.
It is called \emph{regular} if under any set of results $V = \left\{ v^1, v^2, \dots , v^k \right\}$ in the preliminary round such that $v^i(x,y) \neq \otimes$ for all $x,y \in X^i$ and $1 \leq i \leq k$, the following conditions hold:
\begin{enumerate}[label=\emph{\alph*)}]
\item \label{Con_a}
$\cup_{j=1}^\ell Y^j \subseteq \cup_{i=1}^{k} X^i$;
\item \label{Con_b}
there exists a common monotonic ranking $R = R^i$ in each group $(X^i, R^i)$ of the preliminary round $\mathcal{P}$ such that $x \succ_{R(v^i)} y$ and $y \in \cup_{j=1}^\ell Y^j$ imply $x \in \cup_{j=1}^\ell Y^j$ for all $x,y \in X^i$, $1 \leq i \leq k$;
\item \label{Con_c}
$x,y \in X^i \cap Y^j$ implies $w^j(x,y) = v^i(x,y)$, where $w^j$ is the set of results in the main round group $(Y^j, S^j)$
\item \label{Con_d}
$x \in X^i$, $y \in X^h$, $i \neq h$, and $x,y \in Y^j$ imply $w^j(x,y) = \otimes$, where $w^j$ is the set of results in the main round group $(Y^j, S^j)$;
\item \label{Con_e}
there exists a common monotonic ranking $S = S^j$ in each group $(Y^j, S^j)$ of the main round $\mathcal{M}$.
\end{enumerate}
\end{definition}

The idea behind a regular tournament with multiple group stages is straightforward. Some top teams from the preliminary round groups qualify for the main round (conditions~\ref{Con_a} and \ref{Con_b}), where they are divided into new groups such that the matches already played against teams in the same main round group are carried over (conditions~\ref{Con_c} and \ref{Con_d}). Furthermore, the rankings in the preliminary and main round groups are monotonic and identical, respectively (conditions~\ref{Con_b} and \ref{Con_e}). Nonetheless, the rankings $R$ and $S$ can be different.

Perhaps these principles have inspired the decision-makers of the EHF.

\begin{definition} \label{Def312}
\emph{Manipulation}:
Let $(\mathcal{P}, \mathcal{M}, \mathcal{Q})$ be a tournament with multiple group stages.  
A team $x \in X^i$ can \emph{manipulate} the tournament if there exist two sets of results $V = \left\{ v^1, v^2, \dots , v^i, \dots , v^k \right\}$ and $\bar{V} = \left\{ v^{1}, v^2, \dots , \bar{v}^i, \dots , v^k \right\}$ in the preliminary round such that $\bar{v}_2^i(x,y) \geq v_2^i(x,y)$ and $\bar{v}_1^i(y,x) \geq v_1^i(y,x)$ for all $y \in X^i$, furthermore, $x \in Y^j$, $1 \leq j \leq \ell$ according to both $\mathcal{Q}(V)$ and $\mathcal{Q}(\bar{V})$, and either $s_{\bar{W}}(x) > s_W(x)$, or $s_{\bar{W}}(x) = s_{W}(x)$ and $gd_{\bar{W}}(x) > gd_{W}(x)$.
\end{definition}

Manipulation means that team $x$ can increase its number of points ($s_{\bar{W}}(x) > s_W(x)$), or at least improve its goal difference ($gd_{\bar{W}}(x) > gd_{W}(x)$) with preserving its number of points ($s_{\bar{W}}(x) = s_{W}(x)$) in the main round by conceding more goals in a match of the preliminary round.
The definition of manipulation may seem to be restrictive but:
(a) scoring fewer goals is not an option at a given standing of the game; and
(b) qualification to another main round group is not necessarily advantageous even if better results are carried over.

Since conceding more goals is in the hands of a team, it can be regarded as its decision variable.

\begin{definition} \label{Def313}
\emph{Strategy-proofness}:
A tournament with multiple group stages $(\mathcal{P}, \mathcal{M}, \mathcal{Q})$ is called \emph{strategy-proof} if there exists no set of group results $V = \left\{ v^1, v^2, \dots ,v^k \right\}$ under which a team $x \in \cup_{i=1}^k X^i$ can manipulate.
\end{definition}

Our central result concerns the strategy-proofness of regular tournaments with multiple group stages: while manipulation certainly worsens a team's goal difference (and sometimes its number of points, too) in its preliminary round group as the ranking rule applied here is monotonic, this might pay off in the main round, where some matches of the preliminary round are discarded---provided that the team still qualifies.

\begin{proposition} \label{Prop31}
Let $(\mathcal{P}, \mathcal{M}, \mathcal{Q})$ be a regular tournament with multiple group stages such that the following conditions hold:
\begin{itemize}
\item
there exist $x,y \in X^i \cap Y^j$ for some $1 \leq i \leq k$ and $1 \leq j \leq \ell$;
\item
there exists $u \in X^i$ but $u \notin Y^j$.
\end{itemize}
Then this tournament with multiple group stages does not satisfy strategy-proofness.
\end{proposition}

According to the conditions of Proposition~\ref{Prop31}, the result of at least one match played in the preliminary round (between the teams $x$ and $y$) is carried over to the main round, and the results of some matches (between the teams $x$ or $y$, and $u$) are ignored in the main round. 

\begin{proof}
It works by simplifying the motivating example of Section~\ref{Sec2}.

\begin{example} \label{Examp31}
Consider a regular tournament with multiple group stages $(\mathcal{P}, \mathcal{M}, \mathcal{Q})$.
Let $(X^1,R)$ be a single round-robin group in the preliminary round with $X^1 = \{ a,b,c \}$. Therefore, the number of points for each team is between $2 \gamma$ and $2 \alpha$.

Assume that there is $\ell = 1$ group in the main round and $x \in X^1 \cap Y^1$ if and only if $\left\{ z \in X^1: x \succ_{R(v^1)} z \right\} \neq \emptyset$, namely, the group winner and the runner-up qualify for the main round from the group $(X^1,R)$.


\begin{table}[t]
\addtocounter{table}{-1}
\begin{threeparttable}
\centering
\caption[The round-robin group $(X^1,R)$ of Example~\ref{Examp31}]{The round-robin group $(X^1,R)$ of Example~\ref{Examp31}}
\label{Table3}
\rowcolors{1}{}{gray!20}
    \begin{tabularx}{\linewidth}{cc CCC CCC >{\bfseries}C} \toprule
    Position   	& Team & $a$     & $b$     & $c$ 	 & GF & GA & GD & Pts \\ \bottomrule \showrowcolors
    1     		& $a$  & ---     & 0-1     & 4-0     & 4  & 1  & $3$  & $\alpha + \gamma$ \\
    2     		& $b$  & 1-0     & ---     & 0-2     & 1  & 2  & $-1$ & $\alpha + \gamma$ \\
    3    	 	& $c$  & 0-4     & 2-0     & ---     & 2  & 4  & $-2$ & $\alpha + \gamma$ \\ \toprule 
    1     		& $a$  & ---     & 0-1     & ---     & 0  & 1  & $-1$ & $\gamma$ \\ \bottomrule
    1*     		& $a$* & ---     & ---     & 2-0*    & 2* & 0* & $2$* & $\alpha$* \\ \toprule
    \end{tabularx}
\begin{tablenotes} \footnotesize
\item
Notes: GF = Goals for; GA = Goals against; GD = Goal difference; Pts = Points. \\
The last but one row contains the group winner's benchmark results that are carried over to the main round. \\
The last row contains the group winner's alternative results that are carried over to the main round after it manipulates.
\end{tablenotes}
\end{threeparttable}
\end{table}

A possible set of results in the preliminary round group $(X^1,R)$ is shown in Table~\ref{Table3}. Team $a$ is the group-winner due to the best (head-to-head) goal difference (see criterion~\ref{eq2} of a monotonic group ranking method). Furthermore, it is considered with $s_{W}(a) = \gamma$ points in the main round, after discarding its match against team $c$, the last team in the group by criterion~\ref{eq2} of a monotonic group ranking $R$ (see the last but one row of Table~\ref{Table3}).

Let us examine what happens if $\bar{v}^1(a,c) = (2; 0)$. Then teams $a$, $b$, and $c$ remain with $\alpha + \gamma$ points, but they have head-to-head goal differences $+1$, $-1$, and $0$, respectively. Therefore, $a$ is the first and $c$ is the second according to criterion~\ref{eq2} of the monotonic group ranking $R$, and team $a$ is considered with $s_{\bar{W}}(a) = \alpha > \gamma = s_{W}(a)$ points in the main round as the last row of Table~\ref{Table3} shows.

To conclude, team $a$ has an opportunity to manipulate this regular tournament with multiple group stages under the set of group results $V$, hence it violates strategy-proofness.
\end{example}

Example~\ref{Examp31} contains only three teams, which is minimal under the conditions of Proposition~\ref{Prop31}.
The number of teams can be increased without changing the essence of the counterexample if we add some teams such that all of them have suffered a defeat of 1-0 to teams $a$, $b$, and $c$. Groups can be double round-robin tournaments instead of single ones by copying the game results above. Since a tournament is incentive incompatible if there exists a single group with the threat of manipulation, an arbitrary number of groups can be added to the example.
\end{proof}

Proposition~\ref{Prop31} remains valid if draws are allowed in a tournament with multiple group stages.

\begin{remark} \label{Rem32}
The 11th European Men's Handball Championship (EHF Euro 2014), discussed in Section~\ref{Sec2}, fits into the model presented above. The number of groups in the preliminary round is $k=4$, the number of groups in the main round is $\ell=2$, and it is a regular tournament with multiple group stages:
\begin{enumerate}[label=\emph{\alph*)}]
\item
$Y^1 \subset X^1 \cup X^2$ and $Y^2 \subset X^3 \cup X^4$;
\item
Ranking in the preliminary round groups is monotonic as it is based on the number of points with tie-breaking through head-to-head results, and the top three teams qualify for the main round;
\item
Matches played in the preliminary round against opponents which qualified to the main round are kept and remain valid for the ranking of the main round;
\item
In the main round, each team faces three teams that did not participate in its preliminary round group;
\item
Ranking in the main round groups is monotonic as it is based on the number of points with tie-breaking through head-to-head results.
\end{enumerate}
\end{remark}

\begin{example} \label{Examp32}
The 11th European Men's Handball Championship (EHF Euro 2014) is not strategy-proof.
\end{example}

\begin{proof}
The scenario presented in Section~\ref{Sec2} shows that team $\text{Poland} = x \in X^3$ can manipulate against $\text{Russia} = y \in X^3$: there exist two sets of group results $V = \left\{ v^1, v^2, v^3, v^4 \right\}$ and $\bar{V} = \left\{ v^1, v^2, \bar{v}^3, v^4 \right\}$ such that $\bar{v}^3 = v^3$ except for $\bar{v}_1^3(x,y) = \bar{v}_2^3(y,x) = 26 > 24 = v_1^3(x,y) = v_2^3(y,x)$, furthermore, Poland qualifies for the group $(X^2,S)$ according to both $\mathcal{Q}(V)$ and $\mathcal{Q}(\bar{V})$, whereas $s_{\bar{W}}(x) = 2 > 0 = s_W(x)$.

Proposition~\ref{Prop31} can also be applied due to Remark~\ref{Rem32}.
\end{proof}

Now we state a positive result, a ``pair'' of Proposition~\ref{Prop31}.

\begin{proposition} \label{Prop32}
Let $(\mathcal{P}, \mathcal{M}, \mathcal{Q})$ be a regular tournament with multiple group stages such that one of the following conditions hold:
\begin{itemize}
\item
there does not exist $x,y \in X^i \cap Y^j$ for any $1 \leq i \leq k$ and $1 \leq j \leq \ell$;
\item
$u,z \in X^i$ and $u \in Y^j$ imply $z \in Y^j$ for all $1 \leq i \leq k$.
\end{itemize}
Then this tournament with multiple group stages is strategy-proof.
\end{proposition}

\begin{proof}
If all preliminary round results achieved against other qualified teams are ignored (first condition), or carried over to the main round (second condition), then it makes no sense to perform weaker in the preliminary round because of the monotonicity of the group rankings in both rounds.
\end{proof}

Proposition~\ref{Prop32} implies that teams qualifying from the same preliminary round group should be drawn into different groups in the main round (which is guaranteed if only one team qualifies from each preliminary round group), or all teams from a given preliminary round group should qualify for the same main round group to avoid incentive incompatibility.

It is also clear from the match discussed in Section~\ref{Sec2} that head-to-head results cannot be used to break a tie in the main round between two teams qualified from the same preliminary round group, otherwise there remain some incentives to influence the set of qualified teams.

Our main result is somewhat related to---but entirely independent of---the finding of \citet{Vong2017} that in general multi-stage tournaments, the necessary and sufficient condition of strategy-proofness is to allow only the top-ranked player to qualify from each group. However, in the model of \citet{Vong2017}, teams tank in order to meet preferred opponents in the next round, thus they only gain in expected value. Contrarily, Definition~\ref{Def313} requires that a team cannot be strictly better off by a lower effort.

\section{A family of incentive incompatible designs} \label{Sec4}

The theoretical results in Section~\ref{Sec3} uncover that there is no straightforward way to guarantee the strategy-proofness of tournaments with multiple group stages and results that are carried over, in contrast to tournament systems consisting of multiple round-robin and knockout tournaments \citep{DagaevSonin2018}, or group-based qualification systems \citep{Csato2020c}.

According to Proposition~\ref{Prop32}, incentive compatibility will be satisfied if either all points scored in the preliminary round are considered in the main round (directly or after an arbitrary monotonic transformation), or all of them are discarded, which is against the essence of these tournaments.
Consequently, the only reasonable solution is to carry over \emph{all} preliminary round results to the main round, perhaps after a monotonic transformation, regardless that some matches were played against teams already eliminated from the tournament.

However, if all results are carried through, then the subsequent phase loses a bit of excitement because there will be greater variation in points at the commencement of this stage, and the teams entering bottom will find it much harder to catch up with the teams entering the stage on top.

This effect can be mitigated by carrying over only half of the points from the preliminary round. The idea comes from the \href{https://en.wikipedia.org/wiki/Belgian_First_Division_A}{Belgian First Division A}, the top league competition for association football clubs in Belgium, where the sixteen participants play a double round-robin tournament in the regular season, followed by a championship play-off for the first six teams such that the points obtained during the regular season are halved. A similar policy is applied currently in the top-tier association football leagues in Poland, Romania, and Serbia \citep{LasekGagolewski2018}. However, in contrast to our model in Section~\ref{Sec3}, the teams advancing to the championship play-off play again in a round-robin format.

For tie-breaking purposes, we suggest retaining the number of goals scored and conceded in the preliminary round. Theoretically, they can also be discarded, but it seems to be unfair when there was a match played in the preliminary round against a team from the same main round group. In the case of Belgian First Division A, goal difference is not among the tie-breaking criteria in the championship play-offs.

Therefore, two incentive compatible versions of each tournament with multiple group stages will be considered without changing the set of matches played:
(1) carrying over all results from the preliminary round; and
(2) carrying over half of the points from the preliminary round.
The consequences of these modifications will be explored here as a kind of cost-benefit analysis via simulations, implemented in the framework of \citet{Csato2021b}. The latter study attempts to compare four tournament formats of the World Men's Handball Championships with respect to several sporting criteria such as selection ability, and competitiveness and quality of the final.

As Table~\ref{Table1} uncovers, the tournament has used three different designs containing multiple group stages. Since the one used in \href{https://en.wikipedia.org/wiki/2003_World_Men\%27s_Handball_Championship}{2003} suffers from various problems and seems not to be efficacious \citep{Csato2021b}, the following two are studied:
\begin{itemize}
\item
Format $G66$:
This design, presented in Figure~\ref{Fig_A2}, has been used first in the \href{https://en.wikipedia.org/wiki/2005_World_Men's_Handball_Championship}{2005 World Men's Handball Championship} and has been applied in the \href{https://en.wikipedia.org/wiki/2009_World_Men's_Handball_Championship}{2009}, \href{https://en.wikipedia.org/wiki/2011_World_Men's_Handball_Championship}{2011}, and \href{https://en.wikipedia.org/wiki/2019_World_Men's_Handball_Championship}{2019} tournaments, too. \\
The preliminary round (see~Figure~\ref{Fig_A2a}) consists of four groups of six teams each such that the top three teams qualify for the main round. The main round consists of two groups of six teams, each created from two preliminary round groups. The top two teams of every main round group advance to the semifinals in the knockout stage (see~Figure~\ref{Fig_A2b}).
The name of the format comes from the size of the groups in its two rounds.

\item
Format $G46$:
This design, presented in Figure~\ref{Fig_A3}, has been used in the \href{https://en.wikipedia.org/wiki/2007_World_Men's_Handball_Championship}{2007 World Men's Handball Championship}. \\
The teams are drawn into six groups of four teams each in the preliminary round (see~Figure~\ref{Fig_A3a}) such that the top two teams proceed to the main round. The main round consists of two groups of six teams, each created from three preliminary round groups. The top four teams of every main round group advance to the quarterfinals in the knockout stage (see~Figure~\ref{Fig_A3b}).\footnote{~The format of the \href{https://en.wikipedia.org/wiki/2020_European_Men's_Handball_Championship}{2020 European Men's Handball Championship} is almost the same as $G46$, the sole difference being that only the two top teams from the two main round groups qualify for the semifinals in the latter tournament.}
Again, the name of the format comes from the size of the groups in its two rounds.
\end{itemize}

While the knockout stage of both tournament formats is immediately determined by the preceding group stages, the competing teams should be drawn into groups before the start of the tournament, thus the seeding regime may also affect the outcome \citep{ScarfYusof2011}. On the other hand, seeding is clearly independent of how the results are carried over from the preliminary round to the main round.

Hence, similarly to \citet{Csato2021b}, two variants of each tournament design, called \emph{seeded} and \emph{unseeded}, are considered.
In the seeded version, the preliminary round groups are drawn such that in the case of groups with $k$ teams ($k=4$ for $G46$ and $k=6$ for $G66$), the strongest $k$ teams are placed in Pot 1, the next strongest $k$ teams in Pot 2, and so on. Then each group gets one team from each pot.
The unseeded version divides the teams into the pots randomly. Therefore, a strong team, allocated in a harsh group, will have more difficulties in qualifying than a ``lucky'' weak team, allocated in an easier group.

\begin{table}[t]
\centering
\caption{Tournament designs considered in the simulations}
\label{Table4}
\rowcolors{1}{}{gray!20}
\begin{tabularx}{0.9\textwidth}{LC ll} \toprule \hiderowcolors
    Notation & Format & Seeding policy & Description \\ \bottomrule \showrowcolors
    $G66$/S & $G66$   & seeded & original incentive incompatible \\
    $G66$/R & $G66$   & unseeded & original incentive incompatible \\
    $G66 \diamond$/S & $G66$   & seeded & all points are carried over \\
    $G66 \diamond$/R & $G66$   & unseeded & all points are carried over \\
    $G66 \star$/S & $G66$   & seeded & half of all points are carried over \\
    $G66 \star$/R & $G66$   & unseeded & half of all points are carried over \\ \bottomrule
    $G46$/S & $G46$   & seeded & original incentive incompatible \\
    $G46$/R & $G46$   & unseeded & original incentive incompatible \\
    $G46 \diamond$/S & $G46$   & seeded & all points are carried over \\
    $G46 \diamond$/R & $G46$   & unseeded & all points are carried over \\
    $G46 \star$/S & $G46$   & seeded & half of all points are carried over \\
    $G46 \star$/R & $G46$   & unseeded & half of all points are carried over \\ \bottomrule
\end{tabularx}
\end{table}

Table~\ref{Table4} summarises the twelve tournament designs to be analysed.

The results of the matches are determined by the \emph{a priori} fixed winning probabilities, which depend on the pre-tournament ranks of the teams $1 \leq i,j \leq 24$, such that a stronger team defeats a weaker team with a higher probability than vice versa.

Further details of the simulation procedure can be found in \citet{Csato2021b}. According to the arguments presented there, all simulations have been implemented with one million runs.

\begin{definition} \label{Def41}
\emph{Tournament metrics}: The tournament designs are compared on the basis of some standard metrics: 
\begin{itemize}
\item
the average pre-tournament rank of the winner, the second-, the third- and the fourth-placed teams;
\item
the expected quality of the final: the sum of the finalists' pre-tournament ranks;
\item
the expected competitive balance of the final: the difference between the finalists' pre-tournament ranks.
\end{itemize}
\end{definition}
Therefore, a lower value of all measures can be preferred.

\begin{figure}[t!]
\centering

\begin{subfigure}{\textwidth}
\caption{Tournament format $G66$}
\label{Fig2a}

\begin{tikzpicture}
\begin{axis}[
name = axis1,
title = Average rank of \#1,
title style = {align=center, font=\small},
width = 0.5\textwidth,
height = 0.3\textwidth,
ymin = 3.25,
ymax = 3.75,
ytick distance = 0.1,
legend style = {font=\small,at={(0.2,-0.15)},anchor=north west,legend columns=6},
symbolic x coords = {A},
xtick = \empty,
ymajorgrids = true,   
ybar = 15pt,
bar width = 15pt,               
]      

\addlegendentry{$G66$/S$\quad$}
\addplot [blue, pattern color = blue, pattern = grid, very thick] coordinates {
(A,3.478447)
};
\addlegendentry{$G66$/R$\quad$}
\addplot [blue, pattern color = blue, pattern = dots, very thick] coordinates {
(A,3.577826)
};
\addlegendentry{$G66 \diamond$/S$\quad$}
\addplot [red, pattern color = red, pattern = grid, very thick] coordinates {
(A,3.353867)
};
\addlegendentry{$G66 \diamond$/R$\quad$}
\addplot [red, pattern color = red, pattern = dots, very thick] coordinates {
(A,3.470588)
};
\addlegendentry{$G66 \star$/S$\quad$}
\addplot [ForestGreen, pattern color = ForestGreen, pattern = grid, very thick] coordinates {
(A,3.424642)
};
\addlegendentry{$G66 \star$/R}
\addplot [ForestGreen, pattern color = ForestGreen, pattern = dots, very thick] coordinates {
(A,3.497109)
};

\legend{}
\end{axis}

\begin{axis}[
at = {(axis1.south east)},
xshift = 0.1\textwidth,
title = Average rank of \#2,
title style = {align=center, font=\small},
width = 0.5\textwidth,
height = 0.3\textwidth,
ymin = 4.55,
ymax = 5.55,
legend style = {font=\small, at={(0,-0.1)},anchor=north west,legend columns=6},
symbolic x coords = {A},
xtick = \empty,
ymajorgrids = true,   
ybar = 15pt,
bar width = 15pt,                 
]      

\addlegendentry{$G66$/S$\quad$}
\addplot [blue, pattern color = blue, pattern = grid, very thick] coordinates {
(A,4.881763)
};
\addlegendentry{$G66$/R$\quad$}
\addplot [blue, pattern color = blue, pattern = dots, very thick] coordinates {
(A,5.101757)
};
\addlegendentry{$G66 \diamond$/S$\quad$}
\addplot [red, pattern color = red, pattern = grid, very thick] coordinates {
(A,4.621875)
};
\addlegendentry{$G66 \diamond$/R$\quad$}
\addplot [red, pattern color = red, pattern = dots, very thick] coordinates {
(A,4.85709)
};
\addlegendentry{$G66 \star$/S$\quad$}
\addplot [ForestGreen, pattern color = ForestGreen, pattern = grid, very thick] coordinates {
(A,4.762881)
};
\addlegendentry{$G66 \star$/R}
\addplot [ForestGreen, pattern color = ForestGreen, pattern = dots, very thick] coordinates {
(A,4.9105)
};

\legend{}
\end{axis}
\end{tikzpicture}

\vspace{0.25cm}
\begin{tikzpicture}
\begin{axis}[
name = axis3,
title = Average rank of \#3,
title style = {align=center, font=\small},
width = 0.5\textwidth,
height = 0.3\textwidth,
ymin = 4.55,
ymax = 5.55,
legend style = {font=\small,at={(0.12,-0.1)},anchor=north west,legend columns=6},
symbolic x coords = {A},
xtick = \empty,
ymajorgrids = true,   
ybar = 15pt,
bar width = 15pt,          
]      

\addlegendentry{$G66$/S$\quad$}
\addplot [blue, pattern color = blue, pattern = grid, very thick] coordinates {
(A,4.936639)
};
\addlegendentry{$G66$/R$\quad$}
\addplot [blue, pattern color = blue, pattern = dots, very thick] coordinates {
(A,5.368936)
};
\addlegendentry{$G66 \diamond$/S$\quad$}
\addplot [red, pattern color = red, pattern = grid, very thick] coordinates {
(A,4.666801)
};
\addlegendentry{$G66 \diamond$/R$\quad$}
\addplot [red, pattern color = red, pattern = dots, very thick] coordinates {
(A,5.118653)
};
\addlegendentry{$G66 \star$/S$\quad$}
\addplot [ForestGreen, pattern color = ForestGreen, pattern = grid, very thick] coordinates {
(A,4.797731)
};
\addlegendentry{$G66 \star$/R}
\addplot [ForestGreen, pattern color = ForestGreen, pattern = dots, very thick] coordinates {
(A,5.172972)
};
\end{axis}
\begin{axis}[
at = {(axis3.south east)},
xshift = 0.1\textwidth,
title = Average rank of \#4,
title style = {align=center, font=\small},
width = 0.5\textwidth,
height = 0.3\textwidth,
ymin = 6.3,
ymax = 8.3,
legend style = {font=\small, at={(-0.5,0)},anchor=north west,legend columns=6},
symbolic x coords = {A},
xtick = \empty,
ymajorgrids = true,   
ybar = 15pt,
bar width = 15pt,     
]      

\addlegendentry{$G66$/S$\quad$}
\addplot [blue, pattern color = blue, pattern = grid, very thick] coordinates {
(A,7.13891)
};
\addlegendentry{$G66$/R$\quad$}
\addplot [blue, pattern color = blue, pattern = dots, very thick] coordinates {
(A,7.87996)
};
\addlegendentry{$G66 \diamond$/S$\quad$}
\addplot [red, pattern color = red, pattern = grid, very thick] coordinates {
(A,6.595498)
};
\addlegendentry{$G66 \diamond$/R$\quad$}
\addplot [red, pattern color = red, pattern = dots, very thick] coordinates {
(A,7.369563)
};
\addlegendentry{$G66 \star$/S$\quad$}
\addplot [ForestGreen, pattern color = ForestGreen, pattern = grid, very thick] coordinates {
(A,6.841237)
};
\addlegendentry{$G66 \star$/R}
\addplot [ForestGreen, pattern color = ForestGreen, pattern = dots, very thick] coordinates {
(A,7.481567)
};
\legend{}
\end{axis}
\end{tikzpicture}

\end{subfigure}

\vspace{0.5cm}
\begin{subfigure}{\textwidth}
\caption{Tournament format $G46$}
\label{Fig2b}

\begin{tikzpicture}
\begin{axis}[
name = axis1,
title = Average rank of \#1,
title style = {align=center, font=\small},
width = 0.5\textwidth,
height = 0.3\textwidth,
ymin = 3.25,
ymax = 3.75,
ytick distance = 0.1,
legend style = {font=\small,at={(0.2,-0.15)},anchor=north west,legend columns=6},
symbolic x coords = {A},
xtick = \empty,
ymajorgrids = true,   
ybar = 15pt,
bar width = 15pt,  
]      

\addlegendentry{$G46$/S$\quad$}
\addplot [blue, pattern color = blue, pattern = grid, very thick] coordinates {
(A,3.609633)
};
\addlegendentry{$G46$/R$\quad$}
\addplot [blue, pattern color = blue, pattern = dots, very thick] coordinates {
(A,3.699409)
};
\addlegendentry{$G46 \diamond$/S$\quad$}
\addplot [red, pattern color = red, pattern = grid, very thick] coordinates {
(A,3.586342)
};
\addlegendentry{$G46 \diamond$/R$\quad$}
\addplot [red, pattern color = red, pattern = dots, very thick] coordinates {
(A,3.678569)
};
\addlegendentry{$G46 \star$/S$\quad$}
\addplot [ForestGreen, pattern color = ForestGreen, pattern = grid, very thick] coordinates {
(A,3.611548)
};
\addlegendentry{$G46 \star$/R}
\addplot [ForestGreen, pattern color = ForestGreen, pattern = dots, very thick] coordinates {
(A,3.680193)
};
\legend{}
\end{axis}

\begin{axis}[
at = {(axis1.south east)},
xshift = 0.1\textwidth,
title = Average rank of \#2,
title style = {align=center, font=\small},
width = 0.5\textwidth,
height = 0.3\textwidth,
ymin = 4.55,
ymax = 5.55,
legend style = {font=\small, at={(0,-0.1)},anchor=north west,legend columns=6},
symbolic x coords = {A},
xtick = \empty,
ymajorgrids = true,   
ybar = 15pt,
bar width = 15pt,                
]      

\addlegendentry{$G46$/S$\quad$}
\addplot [blue, pattern color = blue, pattern = grid, very thick] coordinates {
(A,5.075566)
};
\addlegendentry{$G46$/R$\quad$}
\addplot [blue, pattern color = blue, pattern = dots, very thick] coordinates {
(A,5.309854)
};
\addlegendentry{$G46 \diamond$/S$\quad$}
\addplot [red, pattern color = red, pattern = grid, very thick] coordinates {
(A,5.021251)
};
\addlegendentry{$G46 \diamond$/R$\quad$}
\addplot [red, pattern color = red, pattern = dots, very thick] coordinates {
(A,5.252708)
};
\addlegendentry{$G46 \star$/S$\quad$}
\addplot [ForestGreen, pattern color = ForestGreen, pattern = grid, very thick] coordinates {
(A,5.061443)
};
\addlegendentry{$G46 \star$/R}
\addplot [ForestGreen, pattern color = ForestGreen, pattern = dots, very thick] coordinates {
(A,5.256226)
};

\legend{}
\end{axis}
\end{tikzpicture}

\vspace{0.25cm}
\begin{tikzpicture}
\begin{axis}[
name = axis3,
title = Average rank of \#3,
title style = {align=center, font=\small},
width = 0.5\textwidth,
height = 0.3\textwidth,
ymin = 4.55,
ymax = 5.55,
legend style = {font=\small,at={(0.12,-0.1)},anchor=north west,legend columns=6},
symbolic x coords = {A},
xtick = \empty,
ymajorgrids = true,   
ybar = 15pt,
bar width = 15pt,          
]      

\addlegendentry{$G46$/S$\quad$}
\addplot [blue, pattern color = blue, pattern = grid, very thick] coordinates {
(A,5.187268)
};
\addlegendentry{$G46$/R$\quad$}
\addplot [blue, pattern color = blue, pattern = dots, very thick] coordinates {
(A,5.443733)
};
\addlegendentry{$G46 \diamond$/S$\quad$}
\addplot [red, pattern color = red, pattern = grid, very thick] coordinates {
(A,5.121921)
};
\addlegendentry{$G46 \diamond$/R$\quad$}
\addplot [red, pattern color = red, pattern = dots, very thick] coordinates {
(A,5.381679)
};
\addlegendentry{$G46 \star$/S$\quad$}
\addplot [ForestGreen, pattern color = ForestGreen, pattern = grid, very thick] coordinates {
(A,5.164512)
};
\addlegendentry{$G46 \star$/R}
\addplot [ForestGreen, pattern color = ForestGreen, pattern = dots, very thick] coordinates {
(A,5.388542)
};
\end{axis}
\begin{axis}[
at = {(axis3.south east)},
xshift = 0.1\textwidth,
title = Average rank of \#4,
title style = {align=center, font=\small},
width = 0.5\textwidth,
height = 0.3\textwidth,
ymin = 6.3,
ymax = 8.3,
legend style = {font=\small, at={(-0.5,-0.2)},anchor=north west,legend columns=6},
symbolic x coords = {A},
xtick = \empty,
ymajorgrids = true,   
ybar = 15pt,
bar width = 15pt,     
]      

\addlegendentry{$G46$/S$\quad$}
\addplot [blue, pattern color = blue, pattern = grid, very thick] coordinates {
(A,7.471449)
};
\addlegendentry{$G46$/R$\quad$}
\addplot [blue, pattern color = blue, pattern = dots, very thick] coordinates {
(A,8.020623)
};
\addlegendentry{$G46 \diamond$/S$\quad$}
\addplot [red, pattern color = red, pattern = grid, very thick] coordinates {
(A,7.319355)
};
\addlegendentry{$G46 \diamond$/R$\quad$}
\addplot [red, pattern color = red, pattern = dots, very thick] coordinates {
(A,7.873746)
};
\addlegendentry{$G46 \star$/S$\quad$}
\addplot [ForestGreen, pattern color = ForestGreen, pattern = grid, very thick] coordinates {
(A,7.409856)
};
\addlegendentry{$G46 \star$/R}
\addplot [ForestGreen, pattern color = ForestGreen, pattern = dots, very thick] coordinates {
(A,7.880866)
};
\legend{}
\end{axis}
\end{tikzpicture}
\end{subfigure}

\caption{Expected pre-tournament rank of the first four teams}
\label{Fig2}

\end{figure}


Figure~\ref{Fig2} shows the average pre-tournament rank of the first four teams.
If all points are carried over from the preliminary round, then the result of the tournament becomes more predetermined as the expected rank slightly decreases.
Preserving only half of these points substantially mitigates this loss of excitement, except in the unseeded variant of format $G46$. On the other hand, the average rank of the winner is even higher in the case of seeded $G46$ according to this solution than under the original incentive incompatible design.
In addition, carrying over half of all points minimises the effect of the seeding policy, which seems to be desirable because it is a factor not influenced by the competitors.

\begin{figure}[t!]
\centering

\begin{subfigure}{\textwidth}
\caption{Tournament format $G66$}
\label{Fig3a}

\begin{tikzpicture}
\begin{axis}[
name = axis3,
title = Expected quality,
title style = {align=center, font=\small},
width = 0.5\textwidth,
height = 0.4\textwidth,
ymin = 7.8,
ymax = 9.2,
ytick distance = 0.25,
legend style = {font=\small,at={(0.12,-0.1)},anchor=north west,legend columns=6},
symbolic x coords = {A},
xtick = \empty,
ymajorgrids = true,   
ybar = 15pt,
bar width = 15pt,         
]      

\addlegendentry{$G66$/S$\quad$}
\addplot [blue, pattern color = blue, pattern = grid, very thick] coordinates {
(A,8.36021)
};
\addlegendentry{$G66$/R$\quad$}
\addplot [blue, pattern color = blue, pattern = dots, very thick] coordinates {
(A,8.679583)
};
\addlegendentry{$G66 \diamond$/S$\quad$}
\addplot [red, pattern color = red, pattern = grid, very thick] coordinates {
(A,7.975742)
};
\addlegendentry{$G66 \diamond$/R$\quad$}
\addplot [red, pattern color = red, pattern = dots, very thick] coordinates {
(A,8.327678)
};
\addlegendentry{$G66 \star$/S$\quad$}
\addplot [ForestGreen, pattern color = ForestGreen, pattern = grid, very thick] coordinates {
(A,8.187523)
};
\addlegendentry{$G66 \star$/R}
\addplot [ForestGreen, pattern color = ForestGreen, pattern = dots, very thick] coordinates {
(A,8.407609)
};
\end{axis}
\begin{axis}[
at = {(axis3.south east)},
xshift = 0.1\textwidth,
title = Expected competitive balance,
title style = {align=center, font=\small},
width = 0.5\textwidth,
height = 0.4\textwidth,
ymin = 3.45,
ymax = 4.15,
legend style = {font=\small, at={(-0.5,-0.1)},anchor=north west,legend columns=6},
symbolic x coords = {A},
xtick = \empty,
ymajorgrids = true,   
ybar = 15pt,
bar width = 15pt,   
]      

\addlegendentry{$G66$/S$\quad$}
\addplot [blue, pattern color = blue, pattern = grid, very thick] coordinates {
(A,3.777522)
};
\addlegendentry{$G66$/R$\quad$}
\addplot [blue, pattern color = blue, pattern = dots, very thick] coordinates {
(A,3.941715)
};
\addlegendentry{$G66 \diamond$/S$\quad$}
\addplot [red, pattern color = red, pattern = grid, very thick] coordinates {
(A,3.569972)
};
\addlegendentry{$G66 \diamond$/R$\quad$}
\addplot [red, pattern color = red, pattern = dots, very thick] coordinates {
(A,3.75017)
};
\addlegendentry{$G66 \star$/S$\quad$}
\addplot [ForestGreen, pattern color = ForestGreen, pattern = grid, very thick] coordinates {
(A,3.698637)
};
\addlegendentry{$G66 \star$/R}
\addplot [ForestGreen, pattern color = ForestGreen, pattern = dots, very thick] coordinates {
(A,3.790699)
};
\legend{}
\end{axis}
\end{tikzpicture}

\end{subfigure}

\vspace{0.5cm}
\begin{subfigure}{\textwidth}
\caption{Tournament format $G46$}
\label{Fig3b}

\begin{tikzpicture}
\begin{axis}[
name = axis3,
title = Expected quality,
title style = {align=center, font=\small},
width = 0.5\textwidth,
height = 0.4\textwidth,
ymin = 7.8,
ymax = 9.2,
ytick distance = 0.25,
legend style = {font=\small,at={(0.12,-0.1)},anchor=north west,legend columns=6},
symbolic x coords = {A},
xtick = \empty,
ymajorgrids = true,   
ybar = 15pt,
bar width = 15pt,         
]      

\addlegendentry{$G46$/S$\quad$}
\addplot [blue, pattern color = blue, pattern = grid, very thick] coordinates {
(A,8.685199)
};
\addlegendentry{$G46$/R$\quad$}
\addplot [blue, pattern color = blue, pattern = dots, very thick] coordinates {
(A,9.009263)
};
\addlegendentry{$G46 \diamond$/S$\quad$}
\addplot [red, pattern color = red, pattern = grid, very thick] coordinates {
(A,8.607593)
};
\addlegendentry{$G46 \diamond$/R$\quad$}
\addplot [red, pattern color = red, pattern = dots, very thick] coordinates {
(A,8.931277)
};
\addlegendentry{$G46 \star$/S$\quad$}
\addplot [ForestGreen, pattern color = ForestGreen, pattern = grid, very thick] coordinates {
(A,8.672991)
};
\addlegendentry{$G46 \star$/R}
\addplot [ForestGreen, pattern color = ForestGreen, pattern = dots, very thick] coordinates {
(A,8.936419)
};
\end{axis}
\begin{axis}[
at = {(axis3.south east)},
xshift = 0.1\textwidth,
title = Expected competitive balance,
title style = {align=center, font=\small},
width = 0.5\textwidth,
height = 0.4\textwidth,
ymin = 3.45,
ymax = 4.15,
legend style = {font=\small, at={(-0.5,-0.1)},anchor=north west,legend columns=6},
symbolic x coords = {A},
xtick = \empty,
ymajorgrids = true,   
ybar = 15pt,
bar width = 15pt,   
]      

\addlegendentry{$G46$/S$\quad$}
\addplot [blue, pattern color = blue, pattern = grid, very thick] coordinates {
(A,3.883661)
};
\addlegendentry{$G46$/R$\quad$}
\addplot [blue, pattern color = blue, pattern = dots, very thick] coordinates {
(A,4.086163)
};
\addlegendentry{$G46 \diamond$/S$\quad$}
\addplot [red, pattern color = red, pattern = grid, very thick] coordinates {
(A,3.834947)
};
\addlegendentry{$G46 \diamond$/R$\quad$}
\addplot [red, pattern color = red, pattern = dots, very thick] coordinates {
(A,4.035503)
};
\addlegendentry{$G46 \star$/S$\quad$}
\addplot [ForestGreen, pattern color = ForestGreen, pattern = grid, very thick] coordinates {
(A,3.873039)
};
\addlegendentry{$G46 \star$/R}
\addplot [ForestGreen, pattern color = ForestGreen, pattern = dots, very thick] coordinates {
(A,4.038999)
};
\legend{}
\end{axis}
\end{tikzpicture}
\end{subfigure}

\caption{Characteristics of the tournament final}
\label{Fig3}

\end{figure}

Figure~\ref{Fig3} reinforces these findings by focusing on the final of the tournament: if half of all points scored in the preliminary round are carried over instead of only the results against the teams qualified for the main round, then the final may become a bit more boring but usually involves stronger teams. It decreases the influence of the seeding regime again, especially in the format $G66$.

Following \citet{ScarfYusofBilbao2009}, we have made a robustness check by calculating the metrics for more and less competitive tournaments than the baseline version, in the same way as \citet{Csato2021b}. The qualitative results of these simulations coincide with the findings from Figures~\ref{Fig2} and \ref{Fig3}, hence our observations are independent of the distribution of teams' strength.

The comparison of Figures~\ref{Fig2a} and \ref{Fig2b}, as well as Figures~\ref{Fig3a} and \ref{Fig3b}, uncovers that the choice of the tournament format is more important than the effect of how points are carried over to the main round (see the scales on the vertical axis). Since there is no consensus in the former, at least for the Men's (Women's) World Handball Championships, it does not make much sense to dispute the use of the suggested incentive compatible variants of tournaments with multiple group stages on the basis of the tournament metrics considered.

To conclude, the price of guaranteeing incentive compatibility seems to be negligible---at least, compared to other features of the design like the particular tournament format or the seeding policy. We propose to carry over half of the points scored in the preliminary round. This solution has another interesting implication: it minimises the role of seeding (the difference between the seeded and unseeded variants is the smallest among all designs), which can be advantageous because the true ranking of the teams is never known, and misaligned classification usually leads to unfairness.

On the other hand, carrying over all results from the preliminary round (even after a monotonic transformation) means that the outcome in the main round is at least partially dependent on the strength of opponents in the preliminary group, which might raise questions about fairness if these groups are different in terms of their strength. Fortunately, some authors have recently made useful proposals to balance the difficulty levels of the groups \citep{CeaDuranGuajardoSureSiebertZamorano2020, Guyon2015a, LalienaLopez2019}.

Nonetheless, perfect balance can never be realistically achieved. Therefore, we have examined a scenario to reveal how the proposed family of incentive compatible designs affects competitive balance.\footnote{~We are grateful for an anonymous referee who has suggested this analysis.}
It is assumed that there are two strong teams, one identified correctly with the pre-tournament rank 1, while the other is thought to be only the 13th. The winning probabilities against all other teams are computed as before, according to the model of \citet{Csato2021b}. The unseeded tournament formats are uninteresting since the strongest two teams will play against opponents of the same strength on average. However, compared to the correctly identified strong team, the lower-ranked strong team should face a top team (ranked between 1 and 4 in $G66$, or between 1 and 6 in $G46$) in the preliminary round instead of a middle team (ranked between 13 and 16 in $G66$, or between 13 and 18 in $G46$).

\begin{figure}[t]
\centering

\begin{tikzpicture}
\begin{axis}[width = 0.95\textwidth, 
height = 0.6\textwidth,
xmin = 0.5,
xmax = 16.5,
ymajorgrids,
yticklabel style = {scaled ticks = false, /pgf/number format/fixed, /pgf/number format/precision=4},
xlabel = The pre-tournament rank of the team,
xlabel style = {font = \small},
ylabel = Changes in the probability of winning,
xlabel style = {font = \small},
legend style = {at = {(0.5,-0.15)},anchor = north,legend columns = 6,font = \small}
]
\draw[very thick](axis cs:\pgfkeysvalueof{/pgfplots/xmin},0)  -- (axis cs:\pgfkeysvalueof{/pgfplots/xmax},0);
\addlegendentry{$G66 \diamond$/S$\qquad$}
\addplot [black,thick,only marks,mark=otimes*, mark size=2pt] coordinates {
(1,0.005901)
(2,0.002508)
(3,0.001258)
(4,0.000254)
(5,-0.001226)
(6,-0.001749)
(7,-0.001833)
(8,-0.001708)
(9,-0.00138)
(10,-0.001291)
(11,-0.000827)
(12,-0.000927)
(13,0.002624)
(14,-0.000476)
(15,-0.000369)
(16,-0.00026)
(17,-0.000151)
(18,-0.000137)
(19,-0.00007)
(20,-0.000055)
(21,-0.000037)
(22,-0.000023)
(23,-0.000019)
(24,-0.000007)
};
\addlegendentry{$G66 \star$/S$\qquad$}
\addplot [ForestGreen,very thick,only marks,mark=x, mark size=3pt] coordinates {
(1,0.001565)
(2,-0.000279)
(3,-0.001049)
(4,-0.001408)
(5,0.000361)
(6,-0.000484)
(7,-0.000929)
(8,-0.000742)
(9,-0.000447)
(10,-0.000625)
(11,-0.000648)
(12,-0.000495)
(13,0.006438)
(14,-0.000401)
(15,-0.000281)
(16,-0.000221)
(17,-0.000098)
(18,-0.000095)
(19,-0.000053)
(20,-0.000028)
(21,-0.000038)
(22,-0.000021)
(23,-0.000015)
(24,-0.000007)
};
\addlegendentry{$G46 \diamond$/S$\qquad$}
\addplot[red,thick,only marks,mark=diamond*, mark size=3pt] coordinates {
(1,0.000315)
(2,0.000722)
(3,-0.000296)
(4,0.000562)
(5,0.000025)
(6,0.000306)
(7,0.000304)
(8,-0.000055)
(9,-0.000343)
(10,-0.000224)
(11,-0.000505)
(12,-0.00039)
(13,0.000285)
(14,-0.00026)
(15,-0.000186)
(16,-0.000035)
(17,-0.000086)
(18,-0.000026)
(19,-0.00006)
(20,-0.000007)
(21,0.000004)
(22,-0.000022)
(23,-0.000012)
(24,-0.000016)
};
\addlegendentry{$G46 \star$/S}
\addplot[blue,thick,only marks,mark=star, mark size=3pt] coordinates {
(1,-0.000198)
(2,-0.000253)
(3,0.000202)
(4,-0.000225)
(5,-0.000745)
(6,-0.000272)
(7,0.000775)
(8,0.000259)
(9,-0.000053)
(10,-0.000026)
(11,-0.000241)
(12,-0.00021)
(13,0.001294)
(14,-0.000008)
(15,-0.000091)
(16,-0.000068)
(17,-0.000024)
(18,-0.000022)
(19,-0.000032)
(20,-0.000025)
(21,0)
(22,-0.000011)
(23,-0.000016)
(24,-0.00001)
};
\end{axis}
\end{tikzpicture}

\caption{Changes in the winning probability due to the strategy-proof mechanisms}
\label{Fig4}

\end{figure}


Figure~\ref{Fig4} presents the effect of our proposals on the winning probability of the top 16 teams. Crucially, the impact is marginal, for instance, the choice of the tournament format ($G66$ or $G46$) has a much greater role. Carrying over all points from the preliminary round somewhat worsens competitive balance in design $G66$ as it favours only the five best teams, including the one mistakenly ranked 13th. However, the comparison of the two strong teams (1 and 13) reveals that carrying over only half of all points somewhat reduces the inequality in group strength. The impact is even less significant for the tournament format $G46$.
Hence, organisers should not worry about the unfairness caused by this family of incentive compatible mechanisms, which guarantees an important theoretical requirement by excluding any instances where a team might be better off by losing.

\section{Discussion} \label{Sec5}

Tournament design is an important topic of economics and operations research \citep{Csato2021a}. We have argued that organisers should not miss analysing incentive compatibility because a sporting contest is supposed to be genuine, and is sold to the public as having full integrity. While the actual probability of misaligned incentives can be relatively small, and the audience does not necessarily recognise the problem, it is not worth risking a potential scandal with enormous financial and reputational costs. 
According to our simulation model, the price of guaranteeing the incentive compatibility of tournaments with multiple group stages is marginal: the use of a fair mechanism essentially does not affect the selective ability and the competitive balance of these tournaments.

Somewhat surprisingly, we have not found any controversy about the particular handball match presented in Section~\ref{Sec2}. Nonetheless, its detection is non-trivial as compared to the football and basketball matches discussed in Section~\ref{Sec1} because it was enough to make some mistakes in defence or attack, without the need to score own goals. Reasonably, the EHF remained silent on this issue, and the audience obviously did not study the tie-breaking rules carefully. On the other hand, the Polish coach and players probably knew that they should not make great efforts to win by a higher margin. Hopefully, our paper will contribute to placing this game in the category of the notorious ``\href{https://de.wikipedia.org/wiki/Nichtangriffspakt_von_Gij\%C3\%B3n}{Nichtangriffspakt (Schande) von Gij\'on}''\footnote{~\citet{KendallLenten2017} use the term ``Shame of Gij\'on'', while Wikipedia calls it ``\href{https://en.wikipedia.org/wiki/Disgrace_of_Gij\%C3\%B3n}{Disgrace of Gij\'on}''. The name is given to a 1982 FIFA World Cup football match played between West Germany and Austria at Gij\'on, Spain, on 25 June 1982. A win by one or two goals for West Germany would result in both them and Austria qualifying at the expense of Algeria. West Germany took the lead after 10 minutes, and the remaining 80 minutes were characterised by few serious attempts by either side to score. Both teams were accused of match-fixing although FIFA ruled that they did not break any rules.} \citep[Section~3.9.1]{KendallLenten2017} in the history of sports.
A match played by Australia and West Indies in the \href{https://en.wikipedia.org/wiki/1999_Cricket_World_Cup}{1999 Cricket World Cup} might be an example of similar tacit collusion or emerging cooperation, too \citep[Section~3.7.2]{KendallLenten2017}. However, in contrast to the scenario presented in Section~\ref{Sec2}, this plan---if there was one---did not work out entirely.

Several directions remain open for future research.
First, by the quantification of team strengths and the modelling of match outcomes, the probability of situations susceptible to manipulation can be estimated \citep{ChaterArrondelGayantLaslier2021, Csato2021i, Guyon2020a}.
Second, strategy-proofness can be considered as another aspect in the comparison of different league formats \citep{GoossensBelienSpieksma2012}.
Third, it is clear that there are various trade-offs between efficiency and fairness, and sports administrators implicitly seem to accept some minimal level of tanking \citep{Pauly2014}.
Thus the final aim may be an extensive axiomatic discussion and comparison of sports ranking rules, which has started recently.


\section*{Acknowledgements}
\addcontentsline{toc}{section}{Acknowledgements}
\noindent
This paper could not have been written without \emph{my father} (also called \emph{L\'aszl\'o Csat\'o}), who has coded the simulation in Python. \\
We are grateful to \emph{Liam Lenten} and \emph{Tam\'as Halm} for useful advice. \\
Fourteen anonymous reviewers provided valuable comments and suggestions on earlier drafts. \\
We are indebted to the \href{https://en.wikipedia.org/wiki/Wikipedia_community}{Wikipedia community} for collecting and structuring invaluable information on the sports tournaments discussed. \\
The research was supported by the MTA Premium Postdoctoral Research Program grant PPD2019-9/2019.

\bibliographystyle{apalike}
\bibliography{All_references}

\clearpage

\section*{Appendix}
\addcontentsline{toc}{section}{Appendix}

\renewcommand\thefigure{A.\arabic{figure}}
\setcounter{figure}{0}

\makeatletter
\renewcommand\p@subfigure{A.\arabic{figure}}
\makeatother

\begin{figure}[ht!]
\centering

\begin{subfigure}{\textwidth}
  \centering
  \subcaption{Group stages: preliminary and main rounds}
  \label{Fig_A1a}

\begin{tikzpicture}[scale=1, auto=center, transform shape, >=triangle 45]
  \path
    (-6,0) coordinate (A) node {
    \begin{tabular}{c} \toprule
    \textbf{Group A} \\ \midrule
    A1 \\
    A2 \\
    A3 \\ \midrule 
    A4 \\ \bottomrule
    \end{tabular}
    }
    (6,0) coordinate (B) node { 
    \begin{tabular}{c} \toprule
    \textbf{Group B} \\ \midrule
    B1 \\
    B2 \\
    B3 \\ \midrule 
    B4 \\ \bottomrule
    \end{tabular}
    }
    (0,0) coordinate (X) node { 
    \begin{tabular}{c} \toprule
    \textbf{Group X} \\ \midrule
    X1 \\
    X2 \\ \midrule
    X3 \\  
    X4 \\
    X5 \\
    X6 \\ \bottomrule
    \end{tabular}
    };

  \draw[->] (-5,0.5)--(-1.25,1.75);
  \draw[->] (-5,0)--(-1.25,1.75);
  \draw[->] (-5,-0.5)--(-1.25,1.75);
  \draw[->] (5,0.5)--(1.25,1.75);
  \draw[->] (5,0)--(1.25,1.75);
  \draw[->] (5,-0.5)--(1.25,1.75);
\end{tikzpicture}

\begin{tikzpicture}[scale=1, auto=center, transform shape, >=triangle 45]
  \path
    (-6,0) coordinate (C) node {
    \begin{tabular}{c} \toprule
    \textbf{Group C} \\ \midrule
    C1 \\
    C2 \\
    C3 \\ \midrule 
    C4 \\ \bottomrule
    \end{tabular}
    }
    (6,0) coordinate (D) node { 
    \begin{tabular}{c} \toprule
    \textbf{Group D} \\ \midrule
    D1 \\
    D2 \\
    D3 \\ \midrule 
    D4 \\ \bottomrule
    \end{tabular}
    }
    (0,0) coordinate (Y) node { 
    \begin{tabular}{c} \toprule
    \textbf{Group Y} \\ \midrule
    Y1 \\
    Y2 \\ \midrule
    Y3 \\  
    Y4 \\
    Y5 \\
    Y6 \\ \bottomrule
    \end{tabular}
    };

  \draw[->] (-5,0.5)--(-1.25,1.75);
  \draw[->] (-5,0)--(-1.25,1.75);
  \draw[->] (-5,-0.5)--(-1.25,1.75);
  \draw[->] (5,0.5)--(1.25,1.75);
  \draw[->] (5,0)--(1.25,1.75);
  \draw[->] (5,-0.5)--(1.25,1.75);
\end{tikzpicture}
\end{subfigure}

\vspace{2cm}
\begin{subfigure}{\textwidth}
  \centering
  \subcaption{Knockout stage}
  \label{Fig_A1b}
  
\begin{tikzpicture}[
  level distance=5cm,every node/.style={minimum width=2cm,inner sep=0pt},
  edge from parent/.style={ultra thick,draw},
  level 1/.style={sibling distance=4cm},
  level 2/.style={sibling distance=2cm},
  legend/.style={inner sep=3pt}
]
\node (1) {\Pair{F}{$\mathcal{W}$/SF1}{$\mathcal{W}$/SF2}}
[edge from parent fork left,grow=left]
child {node (2) {\Pair{SF1}{X1}{Y2}}
}
child {node {\Pair{SF2}{X2}{Y1}}
};
\node[legend] at ([yshift=1cm]2) (SF) {\textbf{Semifinals}};
\node[legend] at (1|-SF) (SF) {\textbf{Final}};
\node[legend] at ([yshift=-1cm]1) {\textbf{Third place}};
\node[legend] at ([yshift=-2cm]1) {{\Pair{BM}{$\mathcal{L}$/SF1}{$\mathcal{L}$/SF2}}};
\end{tikzpicture}
\end{subfigure}

\caption{The tournament format of the \href{https://en.wikipedia.org/wiki/2014_European_Men\%27s_Handball_Championship}{2014 European Men's Handball Championship}}
\label{Fig_A1}

\end{figure}

\begin{figure}[ht!]
\centering

\begin{subfigure}{\textwidth}
  \centering
  \subcaption{Group stages: preliminary and main rounds}
  \label{Fig_A2a}

\begin{tikzpicture}[scale=1, auto=center, transform shape, >=triangle 45]
  \path
    (-6,0) coordinate (A) node {
    \begin{tabular}{c} \toprule
    \textbf{Group A} \\ \midrule
    A1 \\
    A2 \\
    A3 \\ \midrule 
    A4 \\
    A5 \\
    A6 \\ \bottomrule
    \end{tabular}
    }
    (6,0) coordinate (B) node { 
    \begin{tabular}{c} \toprule
    \textbf{Group B} \\ \midrule
    B1 \\
    B2 \\
    B3 \\ \midrule 
    B4 \\
    B5 \\
    B6 \\ \bottomrule
    \end{tabular}
    }
    (0,0) coordinate (X) node { 
    \begin{tabular}{c} \toprule
    \textbf{Group X} \\ \midrule
    X1 \\
    X2 \\ \midrule
    X3 \\  
    X4 \\
    X5 \\
    X6 \\ \bottomrule
    \end{tabular}
    };

  \draw[->] (-5,1)--(-1.25,1.75);
  \draw[->] (-5,0.5)--(-1.25,1.75);
  \draw[->] (-5,0)--(-1.25,1.75);
  \draw[->] (5,1)--(1.25,1.75);
  \draw[->] (5,0.5)--(1.25,1.75);
  \draw[->] (5,0)--(1.25,1.75);
\end{tikzpicture}

\begin{tikzpicture}[scale=1, auto=center, transform shape, >=triangle 45]
  \path
    (-6,0) coordinate (C) node {
    \begin{tabular}{c} \toprule
    \textbf{Group C} \\ \midrule
    C1 \\
    C2 \\
    C3 \\ \midrule 
    C4 \\
    C5 \\
    C6 \\ \bottomrule
    \end{tabular}
    }
    (6,0) coordinate (D) node { 
    \begin{tabular}{c} \toprule
    \textbf{Group D} \\ \midrule
    D1 \\
    D2 \\
    D3 \\ \midrule 
    D4 \\
    D5 \\
    D6 \\ \bottomrule
    \end{tabular}
    }
    (0,0) coordinate (Y) node { 
    \begin{tabular}{c} \toprule
    \textbf{Group Y} \\ \midrule
    Y1 \\
    Y2 \\ \midrule
    Y3 \\  
    Y4 \\
    Y5 \\
    Y6 \\ \bottomrule
    \end{tabular}
    };

  \draw[->] (-5,1)--(-1.25,1.75);
  \draw[->] (-5,0.5)--(-1.25,1.75);
  \draw[->] (-5,0)--(-1.25,1.75);
  \draw[->] (5,1)--(1.25,1.75);
  \draw[->] (5,0.5)--(1.25,1.75);
  \draw[->] (5,0)--(1.25,1.75);
\end{tikzpicture}
\end{subfigure}

\vspace{1cm}
\begin{subfigure}{\textwidth}
  \centering
  \subcaption{Knockout stage}
  \label{Fig_A2b}
  
\begin{tikzpicture}[
  level distance=5cm,every node/.style={minimum width=2cm,inner sep=0pt},
  edge from parent/.style={ultra thick,draw},
  level 1/.style={sibling distance=4cm},
  level 2/.style={sibling distance=2cm},
  legend/.style={inner sep=3pt}
]
\node (1) {\Pair{F}{$\mathcal{W}$/SF1}{$\mathcal{W}$/SF2}}
[edge from parent fork left,grow=left]
child {node (2) {\Pair{SF1}{X1}{Y2}}
}
child {node {\Pair{SF2}{X2}{Y1}}
};
\node[legend] at ([yshift=1cm]2) (SF) {\textbf{Semifinals}};
\node[legend] at (1|-SF) (SF) {\textbf{Final}};
\node[legend] at ([yshift=-1cm]1) {\textbf{Third place}};
\node[legend] at ([yshift=-2cm]1) {{\Pair{BM}{$\mathcal{L}$/SF1}{$\mathcal{L}$/SF2}}};
\end{tikzpicture}
\end{subfigure}

\caption{Tournament format $G66$ of the \href{https://en.wikipedia.org/wiki/2011_World_Men\%27s_Handball_Championship}{2011} and the \href{https://en.wikipedia.org/wiki/2019_World_Men\%27s_Handball_Championship}{2019 World Men's Handball Championship}s}
\label{Fig_A2}

\end{figure}

\begin{figure}[ht!]
\centering

\begin{subfigure}{\textwidth}
  \centering
  \subcaption{Group stages: preliminary and main rounds}
  \label{Fig_A3a}

\begin{tikzpicture}[scale=1, auto=center, transform shape, >=triangle 45]
  \path
    (-6,3.5) coordinate (A) node {
    \begin{tabular}{c} \toprule
    \textbf{Group A} \\ \midrule
    A1 \\
    A2 \\ \midrule
    A3 \\  
    A4 \\ \bottomrule
    \end{tabular}
    }
    (-6,0) coordinate (B) node { 
    \begin{tabular}{c} \toprule
    \textbf{Group B} \\ \midrule
    B1 \\
    B2 \\ \midrule
    B3 \\ 
    B4 \\ \bottomrule
    \end{tabular}
    }
    (-6,-3.5) coordinate (C) node {
    \begin{tabular}{c} \toprule
    \textbf{Group C} \\ \midrule
    C1 \\
    C2 \\ \midrule
    C3 \\ 
    C4 \\ \bottomrule
    \end{tabular}
    }
    (6,3.5) coordinate (D) node {
    \begin{tabular}{c} \toprule
    \textbf{Group D} \\ \midrule
    D1 \\
    D2 \\ \midrule
    D3 \\  
    D4 \\ \bottomrule
    \end{tabular}
    }
    (6,0) coordinate (E) node { 
    \begin{tabular}{c} \toprule
    \textbf{Group E} \\ \midrule
    E1 \\
    E2 \\ \midrule
    E3 \\ 
    E4 \\ \bottomrule
    \end{tabular}
    }
    (6,-3.5) coordinate (F) node {
    \begin{tabular}{c} \toprule
    \textbf{Group F} \\ \midrule
    F1 \\
    F2 \\ \midrule
    F3 \\ 
    F4 \\ \bottomrule
    \end{tabular}
    }
    (0,2.5) coordinate (X) node { 
    \begin{tabular}{c} \toprule
    \textbf{Group X} \\ \midrule
    X1 \\
    X2 \\
    X3 \\  
    X4 \\ \midrule
    X5 \\
    X6 \\ \bottomrule
    \end{tabular}
    }
    (0,-2.5) coordinate (Y) node { 
    \begin{tabular}{c} \toprule
    \textbf{Group Y} \\ \midrule
    Y1 \\
    Y2 \\
    Y3 \\  
    Y4 \\ \midrule
    Y5 \\
    Y6 \\ \bottomrule
    \end{tabular}
    };
  \draw[->] (-5,4)--(-1.25,4.25);
  \draw[->] (-5,3.5)--(-1.25,4.25);
  \draw[->] (-5,0.5)--(-1.25,4.25);
  \draw[->] (-5,0)--(-1.25,4.25);
  \draw[->] (-5,-3)--(-1.25,4.25);
  \draw[->] (-5,-3.5)--(-1.25,4.25);
  \draw[->] (5,4)--(1.25,-0.75);
  \draw[->] (5,3.5)--(1.25,-0.75);
  \draw[->] (5,0.5)--(1.25,-0.75);
  \draw[->] (5,0)--(1.25,-0.75);
  \draw[->] (5,-3)--(1.25,-0.75);
  \draw[->] (5,-3.5)--(1.25,-0.75);
\end{tikzpicture}
\end{subfigure}

\vspace{1cm}
\begin{subfigure}{\textwidth}
  \centering
  \subcaption{Knockout stage}
  \label{Fig_A3b}
  
  \begin{tikzpicture}[
  level distance=4cm,every node/.style={minimum width=2cm,inner sep=0pt},
  edge from parent/.style={ultra thick,draw},
  level 1/.style={sibling distance=4cm},
  level 2/.style={sibling distance=2cm},
  legend/.style={inner sep=3pt}
]
\node (1) {\Pair{F}{$\mathcal{W}$/SF1}{$\mathcal{W}$/SF2}}
[edge from parent fork left,grow=left]
child {node (2) {\Pair{SF1}{$\mathcal{W}$/QF1}{$\mathcal{W}$/QF2}}
child {node (3) {\Pair{QF1}{X1}{Y4}}
} 
child {node {\Pair{QF2}{X3}{Y2}}
} 
} 
child {node {\Pair{SF2}{$\mathcal{W}$/QF3}{$\mathcal{W}$/QF4}}
child {node {\Pair{QF3}{X2}{Y3}}
} 
child {node {\Pair{QF4}{X4}{Y1}}
} 
};
\node[legend] at ([yshift=1cm]3) (QF) {\textbf{Quarterfinals}};
\node[legend] at (2|-QF) {\textbf{Semifinals}};
\node[legend] at (1|-QF) {\textbf{Final}};
\node[legend] at ([yshift=-1cm]1) {\textbf{Third place}};
\node[legend] at ([yshift=-2cm]1) {{\Pair{BM}{$\mathcal{L}$/SF1}{$\mathcal{L}$/SF2}}};
\end{tikzpicture}
\end{subfigure}

\caption{Tournament format $G46$ of the \href{https://en.wikipedia.org/wiki/2007_World_Men\%27s_Handball_Championship}{2007 World Men's Handball Championship}}
\label{Fig_A3}

\end{figure}


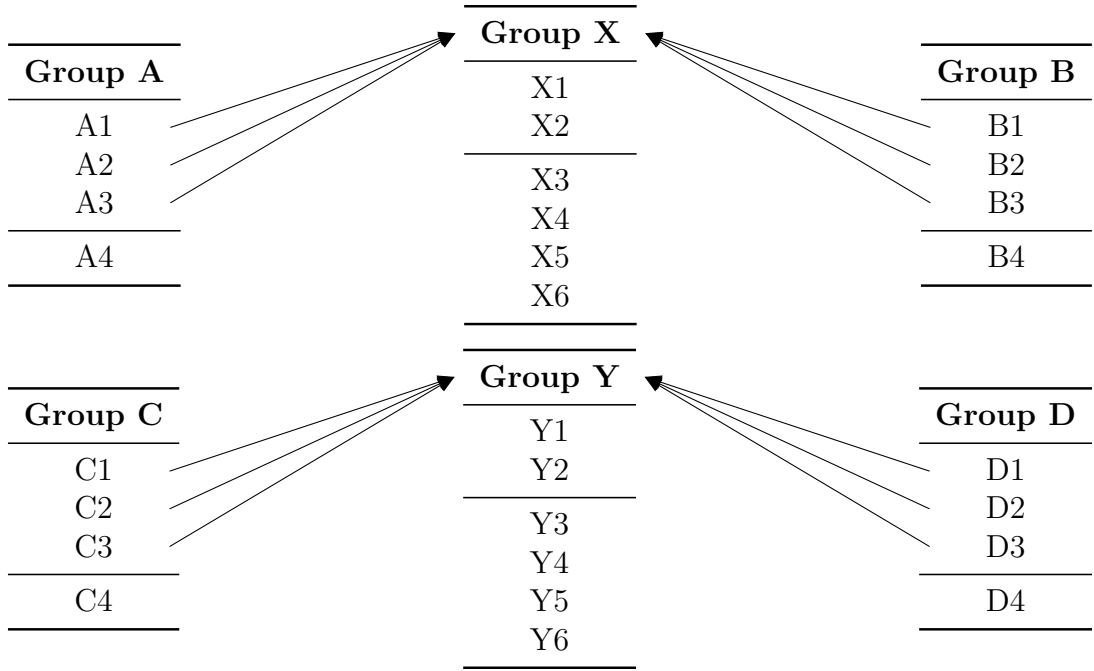
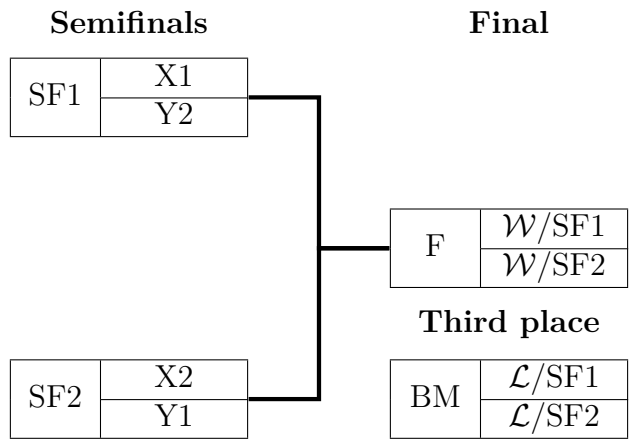
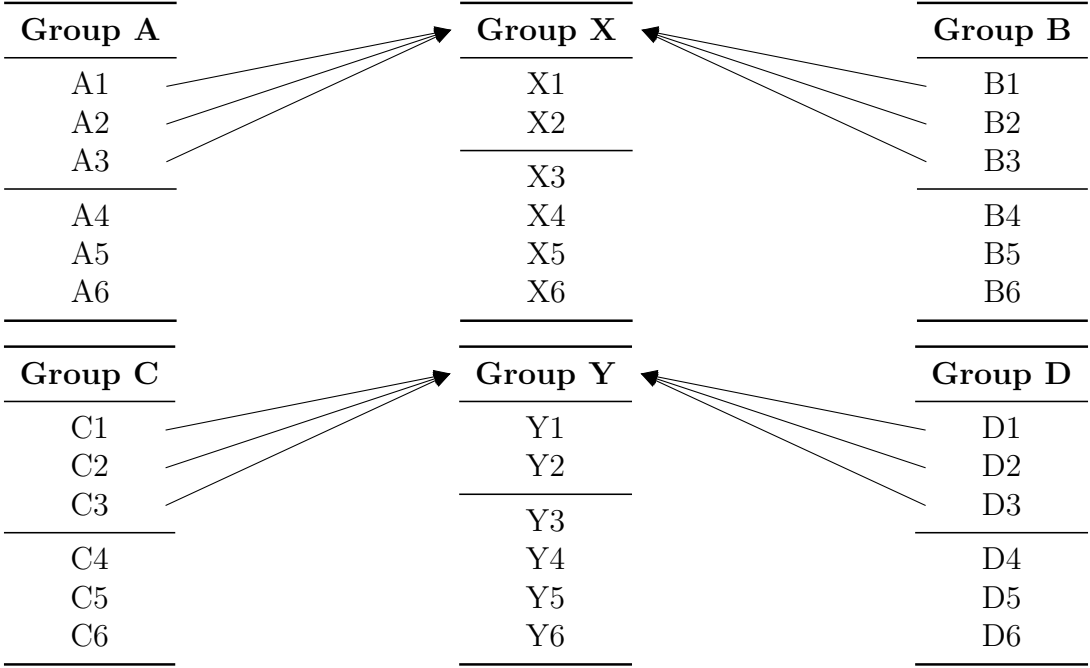
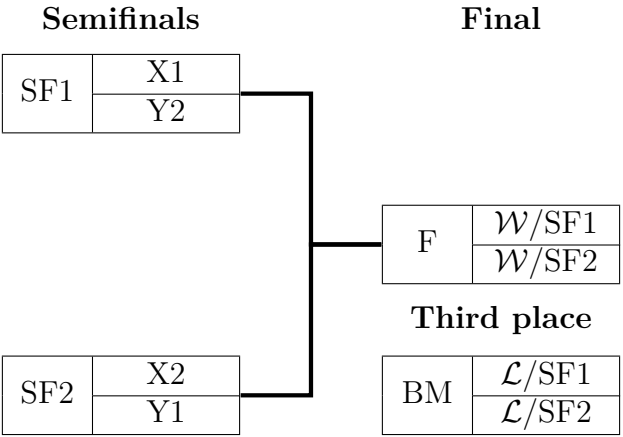
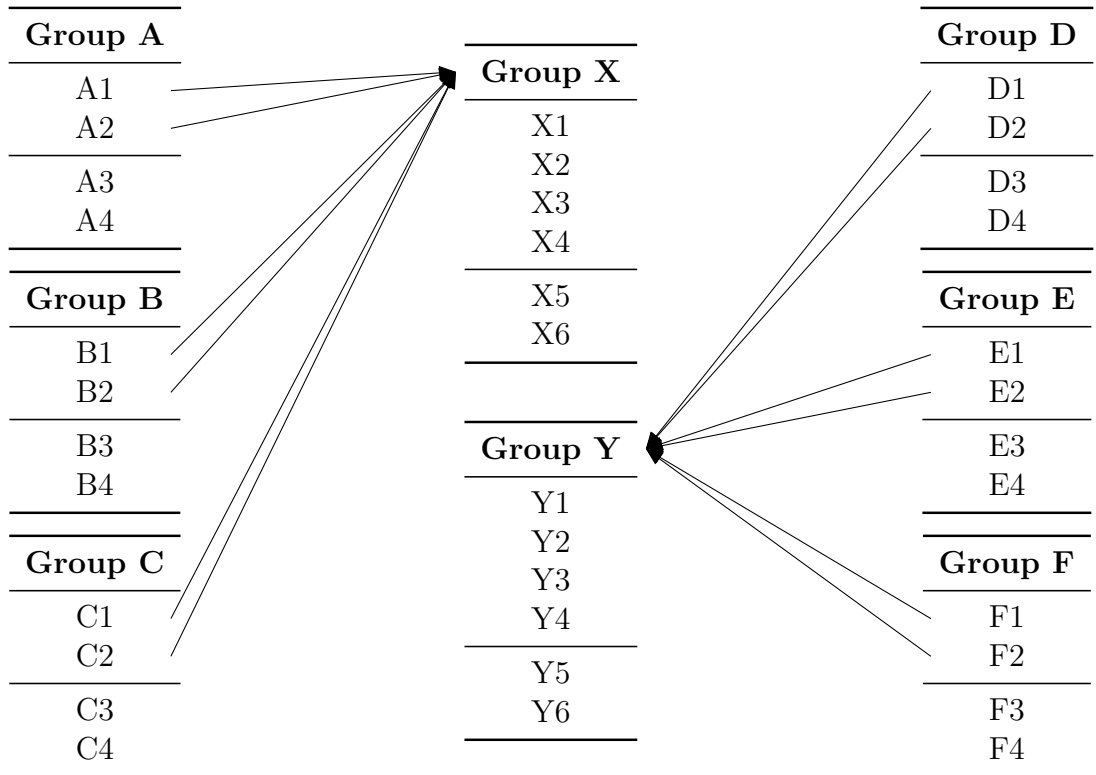
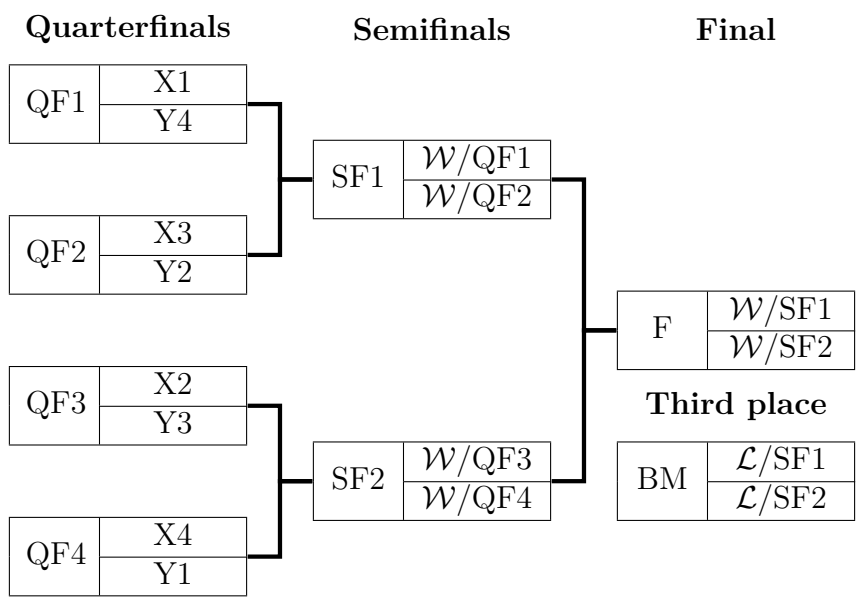
\end{document}